\documentclass[11pt,letterpaper]{article}
\usepackage[letterpaper, margin=0.75in]{geometry}

\usepackage{graphicx,color}
\usepackage[cp850]{inputenc}
\usepackage[T1]{fontenc}
\usepackage{rotating}
\usepackage[f]{esvect}
\usepackage{amsmath, amsthm, amssymb}
\usepackage{rotating}
\usepackage{lscape}
\usepackage{mathrsfs}
\usepackage{subcaption}
\usepackage{mathtools}
\usepackage[toc,page]{appendix}
\usepackage{hyphenat}
\usepackage{mathptmx}
\usepackage[scaled=.65]{helvet}
\usepackage{courier}
\usepackage{rotating}
\usepackage{float}
\usepackage[right,displaymath, mathlines]{lineno}

\usepackage[normalem]{ulem}

\def\tablenotes{\bgroup\parfillskip=0pt plus 1fil
\leftskip=0pt\relax \rightskip=0pt
\vskip2pt\footnotesize}
\def\endtablenotes{\vskip1pt\egroup}

\def\sphline{\noalign{\vskip3pt}\hline\noalign{\vskip3pt}}

\newtheorem{theorem}{Theorem}[section]
\newtheorem{proposition}[theorem]{Proposition}

\newtheorem{lemma}[theorem]{Lemma}

\newtheorem{definition}[theorem]{Definition}
\renewcommand{\epsilon}{\varepsilon}
\renewcommand{\leq}{\leqslant}
\renewcommand{\geq}{\geqslant}
\renewcommand{\d}{\mathrm{d}}

\renewcommand{\epsilon}{\varepsilon}

\newcommand{\TimeDeriv}{\frac{\textrm{d}}{\textrm{dt}}}

\allowdisplaybreaks

\usepackage[square,numbers,sort]{natbib}
\usepackage{graphicx}
\usepackage{overpic}
\usepackage{tikz}
 \graphicspath{ {./Figures/} }
 
 \usepackage{setspace}
 
 \usepackage[right,mathlines]{lineno}
\let\oldequation\equation
\let\oldendequation\endequation

\renewenvironment{equation}
  {\linenomathNonumbers\oldequation}
  {\oldendequation\endlinenomath}
  
\let\oldalign\align
\let\oldendalign\endalign

\renewenvironment{align}
  {\linenomathNonumbers\oldalign}
  {\oldendalign\endlinenomath}

\renewenvironment{align*}
  {\linenomathNonumbers\oldalign\notag}
  {\notag \oldendalign \endlinenomath}   
 
\begin{document}
\title{Inheritance of intracellular viral RNA in a multiscale model of hepatitis C infection}
\author{Tyler Cassidy\textsuperscript{1}, Giulia Belluccini\textsuperscript{2}, Sarafa A. Iyaniwura\textsuperscript{3}, Ruy M. Ribeiro\textsuperscript{2}, and   Alan S. Perelson\textsuperscript{2} }  
\maketitle
\small{ 
\textsuperscript{1}  School of Mathematics, University of Leeds, Leeds, LS2 9JT, United Kingdom.  
\\
\textsuperscript{2} Theoretical Biology and Biophysics, Theoretical Division, Los Alamos National Laboratory, Los Alamos, NM 87545, USA. 
\\
\textsuperscript{3}  Vaccine and Infectious Disease Division, Fred Hutchinson Cancer Center, Seattle, WA, USA.
\maketitle

\section*{Abstract}
Multiscale mathematical models of hepatitis C infection have been instrumental in our understanding of direct acting antivirals. These models include the mechanisms driving intracellular viral production and explicitly model the intracellular concentration of viral RNA. Incorporating proliferation of infected hepatocytes in these models can be subtle, as infected daughter cells inherit viral RNA from the proliferating mother cell. In this note, we show how to incorporate this inheritance within a multiscale model of HCV infection. As in typical multiscale models of HCV infection, we show that this model is mathematically equivalent to a system of ordinary differential equations and perform bifurcation analysis of the resulting ODE that demonstrates that proliferation of infected hepatocytes can lead to infection persistence even if the basic repoductive number is less than one.

\clearpage 

\section{Introduction}  

Mathematical modeling has been instrumental in our understanding of viral kinetics during hepatitis C virus (HCV) infection. For example, viral dynamic modeling has directly linked observed viral kinetics during antiviral treatment with mechanistic insights into the HCV viral lifecycle \citep{Perelson2015,Neumann1998,Perelson2002,Sachithanandham2023}. These viral dynamics models have also identified the specific antiviral effects of direct acting antivirals against HCV and determined the minimum treatment duration to drive cure of HCV infection \citep{Neumann1998,Dahari2007,Guedj2011,Dahari2007a,Snoeck2010,Iwanami2020}. Recent work has coupled the standard viral dynamics model with essential portions of the intracellular viral life cycle in multiscale viral dynamics models \citep{Guedj2013,Quintela2018,Cardozo2020}. These multiscale models typically consist of one or more coupled structured partial differential equations (PDEs) that explicitly model the production of virus within infected hepatocytes with a system of integro-differential equations describing the dynamics of uninfected hepatocytes, infected hepatocytes, and circulating virus. 

These multiscale models can be mathematically complex to analyse and computationally demanding to simulate \citep{Rong2013a,Wang2024,Quintela2018,Rong2013,Guedj2010,Magal2010}. Accordingly, \citet{Kitagawa2018} established the mathematical equivalence between the coupled system of integro-differential equations, for the case in which the model parameters are constants, and a system of ordinary differential equations (ODEs) using techniques similar to the linear chain trick. This equivalence greatly facilitates the use of these multiscale models when analysing clinical data, as the equivalent ODE formulation can be used to estimate model parameters, and has been extensively utilized both in the context of HCV and other viral infections~\citep{Iyaniwura2024,Goncalves2021,Kitagawa2023}. Alternative modeling frameworks, either in the form of highly detailed ODE models of intracellular viral components \citep{Aunins2018,Zitzmann2020,Guedj2010,Dahari2007b} or multiscale models that utilize the intracellular concentration of viral RNA as an independent variable in the resulting PDE \citep{WootdeTrixhe2021,WootdeTrixhe2015} have also been used to understand the dynamics of HCV infection. However, to our knowledge, no existing models capture the inheritance of intracellular viral material upon proliferation of HCV infected hepatocytes, which is particularly relevant for HCV, as it is a completely cytoplasmic virus.

The HCV viral life cycle begins with virus entry and release of the positive strand viral RNA (vRNA) within the cytoplasm of infected hepatocytes. This vRNA is translated and replicated into a negative sense RNA intermediary that permits replication of positive strand vRNA \citep{Li2015b,Quintela2018}. This positive strand vRNA is then assembled into enveloped viral particles that are released into the circulation. Importantly, each stage of the HCV viral replication cycle occurs in the cytoplasm of infected hepatocytes. Consequently, the intracellular vRNA within the cytoplasm is divided during proliferation of infected cells, with daughter cells inheriting intracellular viral material. Including infected cell proliferation in mathematical models that do not capture the dynamics of the intracellular viral life cycle is straightforward, as the daughter cells are treated as being identical to the proliferating mother cell \citep{Snoeck2010,Dahari2009a,Reluga2009,Dahari2007,Dahari2007a,Dahari2005}. Conversely, in the multiscale models that have elucidated the effect of direct acting antivirals, accurately capturing infected cell proliferation is delicate, as these models must explicitly track the inheritance of intracellular virus from mother to daughter cells.  

In recent work, \citet{Elkaranshawy2024} adapted the standard multiscale model of HCV infection introduced by \citet{Guedj2013} to include the proliferation of infected cells. However, \citet{Elkaranshawy2024} worked directly with the equivalent ODE system derived by \citet{Kitagawa2018} rather than considering the underlying multiscale coupled system of PDEs and integro-differential equations. We show in Appendix~\ref{Appendix:Comparison} that this direct adaptation of the ODE model results in the spontaneous generation of intracellular vRNA during proliferation of infected cells thus invalidating the model. Here, we adapt the standard multiscale model of HCV infection\citep{Guedj2013} \citet{Guedj2013} to include infected cell proliferation beginning from the PDE formulation that distinguishes between intracellular and extracellular dynamics. Then, as in the multiscale modeling framework analysed by \citet{Kitagawa2018}, we demonstrate that the resulting multiscale PDE model is equivalent to a system of ODEs under the assumption that all intracellular model parameters are constant functions of infection age. Consequently, we develop a modeling framework that captures both the intracellular viral life cycle and the inheritance of intracellular virus during proliferation. 

In developing this multiscale model, we initially only distinguish between newly infected cells and those that arise from proliferation of an infected mother cell. However, infected hepatocytes arising from proliferation will, on average, inherit half the vRNA from the mother cell before producing vRNA throughout their lifetime. Consequently, intracellular vRNA may accumulate within later generations of infected hepatocytes. We therefore develop a model to track the inheritance of vRNA across generations of infected hepatocytes. The resulting model is an infinite system of coupled partial differential equations that captures generational inheritance of vRNA in a similar manner to \citet{Belluccini2022}. We show that the dynamics of this generational model are identical to those of the simpler modeling framework that only distinguishes between newly infected cells and those arising from proliferation. 

Finally, we study the qualitative dynamics of the multiscale model with proliferation and vRNA inheritance. We calculate the basic reproduction number $\mathcal{R}$ of this multiscale model and give sufficient conditions to ensure that the uninfected equilibrium undergoes the standard forward transcritical bifurcation at $\mathcal{R} =1$. We also demonstrate the existence of proliferation-driven bistability between the uninfected equilibrium and a total-infection equilibrium that is driven by a backwards bifurcation of the uninfected equilibrium. This bifurcation analysis extends earlier work by \citet{Reluga2009}, who identified a backwards bifurcation in the standard viral dynamics model with proliferation of infected cells under quasi-steady state assumptions. We then study the boundary between these forward and backward bifurcations of the uninfected equilibrium by demonstrating the existence of a saddle-node transcritical bifurcation. To our knowledge, this is the first saddle-node transcritical bifurcation identified in a relatively simple model of viral dynamics with only quadratic non-linearities. We illustrate our analytical results using numerical bifurcation analysis in Matcont \citep{Dhooge2008}.

\section{Multiscale model of the HCV viral life cycle}
 
We begin with the standard viral dynamics model that captures the extracellular dynamics in our multiscale model of HCV infection. The standard viral dynamics model tracks the dynamics of uninfected hepatocytes, $T(t)$, infected hepatocytes, $I(t)$, and free virus, $V(t)$ \citep{Perelson1996,Neumann1998}. Typically, uninfected hepatocytes are produced at a constant rate $\lambda$, have a per capita death rate $d_T$, and are infected at a constant rate $\beta $ by free virus. Infected hepatocytes die at per capita rate $\delta$ and produce free virus with a per capita production rate $p$. The free virus is then cleared at a per capita rate $c$, leading to the standard viral dynamics model
\begin{equation}\label{Eq:StandardViralDynamicsODE}
\left \{
\begin{aligned}
\TimeDeriv T(t) & = \lambda - d_T T(t) -\beta V(t) T(t), \\
\TimeDeriv I(t) & = \beta V(t)T(t) - \delta I(t), \\
\TimeDeriv V(t) & = pI(t) - cV(t).
\end{aligned}
\right.
\end{equation}  
This model has been instrumental in our understanding of many viral infections, including HIV-1 and HCV \citep{Perelson2002,Neumann1998,Perelson1996,Perelson2015,Perelson1997}, and forms the backbone of the interactions at the extracellular scale in our multiscale model. 

\subsection{Multiscale model of HCV infection without infected cell proliferation}

The standard viral dynamics model Eq.~\eqref{Eq:StandardViralDynamicsODE} does not account for the intracellular processes leading to the production of virus. Thus, \citet{Guedj2013} introduced a multiscale model to capture the production of vRNA within infected hepatocytes. This model accounts for the infection age, or time since infection of an infected hepatocyte. Then, with $i(t,a)$ denoting the density of infected cells with infection age $a$ at time $t$, the total concentration of infected cells is 
\begin{align*}
 I(t) = \int_0^{\infty} i(t,a)\d a.
\end{align*}
Infected cells are lost at an infection age independent rate $\delta$. We denote the amount of vRNA per infected hepatocyte with infection age $a$ at time $t$ by $r(t,a)$, and assume that vRNA is produced by infected hepatocytes at a constant rate $\alpha.$ Intracellular vRNA degrades at a per capita rate $\mu$ or is assembled into viral particles and exported into the circulation at a per capita rate $\rho$. Thus, the rate at which vRNA is secretedas virions  by infected hepatocytes into the circulation is given by $\rho R(t),$ where $R(t)$ is the total concentration of intracellular vRNA within infected hepatocytes defined by
\begin{align*}
R(t) = \int_{0}^{\infty} r(t,a)i(t,a) \d a.
\end{align*}

The dynamics of uninfected hepatocytes and circulating virus are unchanged from the standard viral dynamics model, so the multiscale model of HCV infection is given by
\begin{equation} \label{Eq:NoProliferationMultiScale}
\left \{ 
\begin{aligned}
\TimeDeriv T(t) & = \lambda - d_T T(t) -\beta V(t) T(t), \\
(\partial_t+ \partial_a)i(t,a) & =  - \delta i(t,a), \\
(\partial_t+ \partial_a)r(t,a) & = \alpha  -(\mu+\rho) r(t,a), \\
\TimeDeriv V(t) & = \int_{0}^{\infty} \rho r(t,a)i(t,a) \d a  - cV(t).
\end{aligned}
\right.
\end{equation}
The boundary conditions of Eq.~\eqref{Eq:NoProliferationMultiScale} correspond to newly infected cells, which are infected hepatocytes with infection age $a =0$ and $\zeta$ copies of intracellular vRNA 
\begin{align*}
i(t,0) = \beta V(t) T(t) \ \textrm{and} \ r(t,0) = \zeta.
\end{align*}
\citet{Guedj2013} assumed that $\zeta = 1$ vRNA/cell, although \citet{Kitagawa2018} considered multiple infection events and allowed $\zeta \geq 1$.
Then, Eq.~\eqref{Eq:NoProliferationMultiScale} has initial conditions
\begin{align*}
V(0) = V_0,\ T(0) = T_0, \ i(0,a) = f_i(a), \ \textrm{and} \ r(0,a) = f_r(a).
\end{align*}
The initial age distributions $f_i$ and $f_r$ are typically taken to be integrable, non-negative functions. We note that, along the characteristics of Eq.~\eqref{Eq:NoProliferationMultiScale}, the concentration of infected hepatocytes decays exponentially so we do not impose a maximal lifespan for these cells. A similar argument shows that the density of vRNA within infected hepatocytes, $r(t,a),$ is bounded as a function of $a$. Consequently, the integral defining $R(t)$ is well-defined. We give a schematic representation of the multiscale model Eq.~\eqref{Eq:NoProliferationMultiScale} in Fig.~\ref{Fig:ModelSchematic}A and B.

\begin{figure}[!htbp]
\noindent
\centering  
 \begin{tabular} {c}
\includegraphics[trim=  20 21 10 10,clip,width=0.95\textwidth]{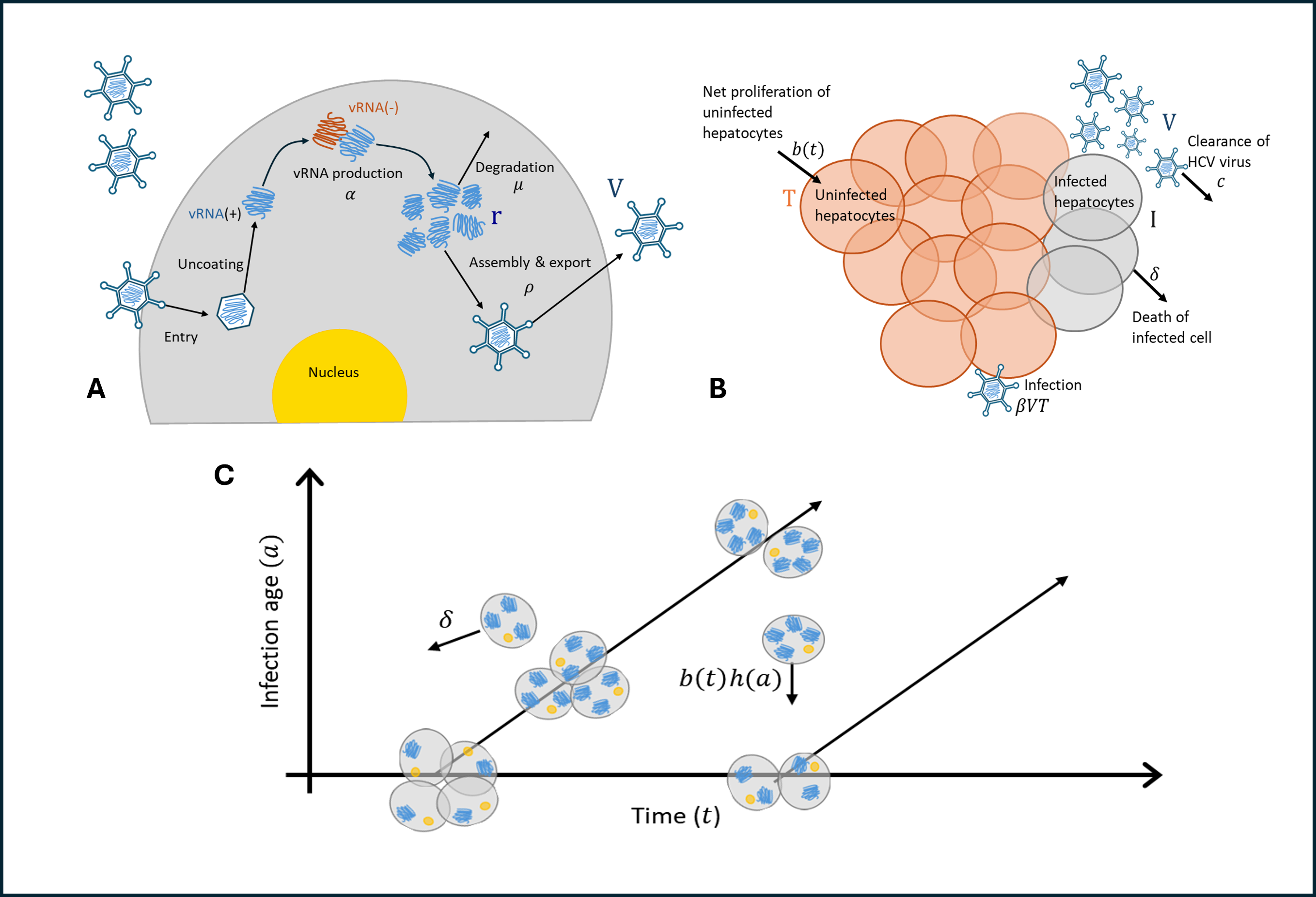} 
\end{tabular}
\caption{ \textbf{Schematic of the multiscale model of HCV infection}.  \textbf{A)} The intracellular HCV life cycle, which begins with infection of a hepatocyte and the release of positive strand vRNA, shown in blue and denoted vRNA(+), into the cell cytoplasm following the uncoating of the HCV viral capsid. This vRNA is translated into viral proteins, such as the RNA-dependent RNA polymerase and other proteins that first copy the positive strand vRNA into its complementary negative strand, shown in orange and denoted vRNA(-), and then form a replication complex, or replication machine, that generates new positive stand vRNA at rate $\alpha$. This positive strand vRNA can be degraded at rate $\mu$, or assembled into HCV particles and exported into the circulation at rate $\rho$. \textbf{B)} The HCV extracellular dynamics, where uninfected hepatocytes $(T)$ are produced due to net proliferation at rate $b(t)$, and become infected cells $(I)$, following infection by virus $(V)$ at rate $\beta$. Infected cells are lost at per capita rate $\delta$ and secrete vRNA containing particles as virus ($V$) into the circulation. The circulating virus is cleared at per capita rate $c$. \textbf{C)} The progression of a cohort of infected hepatocytes with chronological time and infection age on the horizontal and vertical axes, respectively. Upon infection, infected hepatocytes have infection age $a = 0$ and progress through time-age space along the solid lines, which are the characteristic curves of the multiscale PDE. Infected cells are removed from the cohort due to death, which occurs at a constant per capita rate $\delta$, or when proliferating at rate $b(t)h(a)$. An infected hepatocyte with four blue intracellular vRNA molecules proliferates to produce two daughter cells, each with two intracellular vRNA molecules, which appear at the boundary with infection age $a = 0$.}
\label{Fig:ModelSchematic}
\end{figure}

Then, under the assumption that the rates $\delta, \mu, \alpha,$ and $\rho$ are independent of infection age, \citet{Kitagawa2018} used the PDE~\eqref{Eq:NoProliferationMultiScale} and the definitions of  $I(t)$ and $R(t)$ to show that Eq.~\eqref{Eq:NoProliferationMultiScale} is mathematically equivalent to
\begin{equation} \label{Eq:NoProliferationODE}
\left \{
\begin{aligned}
\TimeDeriv T(t) & = \lambda - d_T T(t) -\beta V(t) T(t), \\
\TimeDeriv I(t) & = \beta V(t)T(t) - \delta I(t), \\
\TimeDeriv R(t) & = \zeta \beta V(t) T(t) +  \alpha I(t)  - (\rho+\mu+\delta)R(t), \\
\TimeDeriv V(t) & =  \rho R(t)  - cV(t).
\end{aligned}
\right.
\end{equation}
Here, the initial conditions for $I$ and $R$ arise naturally from the initial densities of the multiscale model Eq.~\eqref{Eq:NoProliferationMultiScale} and are given by 
\begin{align*}
I(0) = \int_0^{\infty} i(0,a) \d a = \int_0^{\infty} f_i(a) \d a \quad \textrm{and} \quad R(0) = \int_0^{\infty} r(0,a)i(0,a) \d a = \int_0^{\infty} f_r(a) f_i(a) \d a.
\end{align*} 
Often, models such as Eq.~\eqref{Eq:NoProliferationODE} are used to describe the viral kinetics in chronically infected participants in clinical trials of novel antivirals that, for example, reduce the rate of vRNA replication, $\alpha$, or the rate of viral production $\rho$ \citep{Guedj2013,Lau2016}. We show how the assumption of chronic infection naturally defines the initial densities $f_{i}(a)$ and $f_{r}(a)$ in Appendix~\ref{Appendix:InitialDensities} and discuss how to incorporate the antiviral effects of direct acting antivirals in Appendix~\ref{Appendix:Treatment}.

\section{Infected cell proliferation in a multiscale model of the HCV life cycle}\label{Sec:MultiscaleProliferation}

We now incorporate cellular proliferation in the multiscale model Eq.~\eqref{Eq:NoProliferationMultiScale}. Rather than including the constant production of uninfected hepatocytes with rate $\lambda$ and the death of uninfected hepatocytes with per capita rate $d_T$ as in Eq.~\eqref{Eq:NoProliferationODE}, we assume that uninfected hepatocytes proliferate with net growth rate $b(t)$ as the immigration of new hepatocytes from precursors is typically slow \citep{Reluga2009,Goyal2017}. This net growth rate $b(t)$ is often chosen to capture logistic growth \citep{Elkaranshawy2024,Goyal2019,Goyal2017,Dahari2009a,Dahari2007a,Dahari2007b}, although other choices are possible. The resulting ODE for uninfected hepatocytes is therefore 
\begin{align*}
\TimeDeriv T(t) & =  b(t) T(t) - \beta V(t) T(t) .
\end{align*}

Next, we consider the proliferation of infected hepatocytes within the multiscale modeling framework before demonstrating how to model cellular inheritance of intracellular vRNA. Here, with the exception of infected cell proliferation, the intra- and extracellular processes leading to vRNA production are unchanged from Eq.~\eqref{Eq:NoProliferationMultiScale} and represented in Fig.~\ref{Fig:ModelSchematic}. 

Now, let $H(t,a)$ be the rate at which infected hepatocytes with infection age $a$ proliferate at time $t$. We assume $H(t,a)$ can be written as the product of the net growth rate of infected cells, denoted by $b(t)$, and the likelihood that a cell with infection age $a$ proliferates, denoted by $h(a)$, so $ H(t,a) = h(a)b(t).$ 
There is \textit{in vitro} evidence that infected hepatocytes may proliferate slower than uninfected cells \citep{Webster2013}. In our multiscale modeling framework, the function $h(a)$ captures this difference in proliferation rates between uninfected and infected hepatocytes, as taking $h(a) < 1$ implies that infected hepatocytes proliferate slower than uninfected cells. Some existing ODE models of HCV infection also sometimes include different proliferation rates for uninfected and infected hepatocytes \citep{Dahari2009a,Dahari2007b,Reluga2009}.

Here, we assume that the probability of a cell with infection age $a$ undergoing proliferation is given by a random variable $\mathbb{A}$. Then, $h(a)$ is the hazard rate defined by $\mathbb{A}$ \citep{Cassidy2018a}. Proliferation results in the removal of the mother cell with infection age $a >0$ and the production of two daughter cells with infection age $a = 0$ for a population net increase of one cell per mitotic event \citep{Cassidy2021}. We illustrate the proliferation of infected hepatocytes within this time-since-infection framework in Fig.~\ref{Fig:ModelSchematic}C.

We distinguish between newly infected cells, denoted by $i_0(t,a)$, and infected cells that arise as the result of proliferation of an infected cell, denoted by $i_p(t,a)$. We assume that the dynamics of these two populations of cells are otherwise identical and satisfy
\begin{align*}
(\partial_t+ \partial_a)i_0(t,a)  =   - [\delta+b(t)h(a)] i_0(t,a) \ \textrm{and} \ 
(\partial_t+ \partial_a)i_p(t,a)  =   - [\delta+b(t)h(a)] i_p(t,a), 
\end{align*}
where we recall that $\delta$ is the infection age independent per capita death rate of infected cells. The total number of infected hepatocytes undergoing proliferation at time $t$ is given by 
\begin{align*}
 \int_0^{\infty} b(t)h(a) i_0(t,a)\d a + \int_0^{\infty} b(t)h(a) i_p(t,a) \d a .
\end{align*}
As proliferation of infected hepatocytes results in daughter cells with infection age $a=0$, these proliferating cells appear in the boundary conditions of the PDEs 
\begin{align*}
i_0(t,0) = \beta V(t) T(t) \quad \textrm{and} \quad i_p(t,0) =  2\int_0^{\infty} b(t)h(a)i_0(t,a)\d a + 2\int_0^{\infty} b(t)h(a)i_p(t,a)\d a.
\end{align*}
 
We now turn to the dynamics of intracellular vRNA within proliferating infected hepatocytes. We once again distinguish between the vRNA within newly infected cells, $r_0(t,a)$, and the vRNA within infected cells that arose from proliferation, $r_p(t,a)$. As before, we assume that intracellular vRNA is assembled into viral particles and exported into the circulation at per capita rate $\rho$ and degrades intracellularly at rate $\mu$. When an infected hepatocyte proliferates, the mother cell is removed from the infected cell density with age $a>0$. However, in the same way that death of infected cells does not decrease the amount of vRNA within the remaining infected cells, this proliferation does \textit{not} result in a corresponding decrease in the amount of vRNA within the remaining cells of age $a$. Consequently, the PDEs describing $r_0(t,a)$ and $r_p(t,a)$ are given by 
\begin{align*}
(\partial_t+ \partial_a)r_0(t,a) = \alpha  - [\mu+\rho] r_0(t,a) \quad \textrm{and} \quad
(\partial_t+ \partial_a)r_p(t,a) = \alpha  - [\mu+\rho] r_p(t,a). 
\end{align*}
\sloppy
The boundary condition for $r_0(t,a)$ corresponds to vRNA arising from infection with $r_0(t,0 ) = \zeta$. Now, the total concentration of vRNA within proliferating infected cells that have not previously proliferated, i.e.~the cells denoted by $i_0(t,a)$ that are proliferating, is given by
\begin{align*}
\int_0^{\infty} r_0(t,a)  b(t)h(a)i_0(t,a) \; \d a ,
\end{align*} 
while the total concentration of vRNA within proliferating infected cells that have previously proliferated, i.e. infected cells $i_p(t,a)$ that are proliferating, is given by 
\begin{align*}
 \int_0^{\infty} r_p(t,a)  b(t)h(a)i_p(t,a) \; \d a.
\end{align*} 
Upon mitosis, the daughter cells inherit, on average, half the intracellular vRNA of the mother cell prior to division. The average intracellular vRNA is given by the total concentration of vRNA within proliferating infected cells divided by the total number of proliferating cells, so 
\begin{align*}
 r_p(t,0) = \frac{1}{2}\left[\frac{ \int_0^{\infty} r_0(t,a)  b(t)h(a)i_0(t,a) \; \d a + \int_0^{\infty} r_p(t,a)  b(t)h(a)i_p(t,a)\; \d a }{\int_0^{\infty} b(t)h(a)i_0(t,a)\; \d a+\int_0^{\infty} b(t)h(a)i_p(t,a)\; \d a} \right].
\end{align*} 
The assumption that each daughter cell inherits half the average vRNA within proliferating cells may appear to over-simplify the potential accumulation of vRNA within infected hepatocytes. Consequently, in Section~\ref{Sec:GenerationTracking}, we extend this model to explicitly track the generational accumulation of vRNA within infected hepatocytes. We then show that the population-level dynamics of the model that tracks generation inheritance are precisely the same as assuming that each daughter cell inherits half the average vRNA within proliferating cells.

Finally, we assume that virus is produced from infected cells, $i_0(t,a)$ and $i_p(t,a)$, by assembly and secretion of their vRNA, $r_0(t,a)$ and $r_p(t,a)$, as virus particles with dynamics that are unchanged from the multiscale model without proliferation in Eq.~\eqref{Eq:NoProliferationMultiScale}. We obtain
\begin{equation} \label{Eq:ProliferationMultiScale}
\left \{  
\begin{aligned}
\TimeDeriv T(t) & =  b(t) T(t)  -\beta V(t) T(t), \\
(\partial_t+ \partial_a)i_0(t,a) & =  - [\delta+b(t)h(a)] i_0(t,a), \\
(\partial_t+ \partial_a)i_p(t,a) & =  - [\delta+b(t)h(a)] i_p(t,a), \\ 
(\partial_t+ \partial_a)r_0(t,a) & = \alpha - [\mu+\rho] r_0(t,a), \\
(\partial_t+ \partial_a)r_p(t,a) & = \alpha  - [\mu+\rho] r_p(t,a), \\
\TimeDeriv V(t) & = \int_{0}^{\infty} \rho \left[ r_0(t,a)i_0(t,a) + r_p(t,a)i_p(t,a)\right]  \d a   - cV(t),
\end{aligned}
\right. 
\end{equation}
with boundary conditions
\begin{equation}\label{Eq:BoundaryConditionsProliferationMultiScale}
\left \{
\begin{aligned}
i_0(t,0) & = \beta V(t) T(t), \\
i_p(t,0) & =  2\int_0^{\infty} b(t)h(a)i_0(t,a)\d a + 2\int_0^{\infty} b(t)h(a)i_p(t,a)\d a, \\
r_0(t,0) & = \zeta, \\
r_p(t,0) & = \frac{1}{2}\left[\frac{ \int_0^{\infty} r_0(t,a)  b(t)h(a)i_0(t,a)\; \d a + \int_0^{\infty} r_p(t,a)  b(t)h(a)i_p(t,a) \; \d a }{\int_0^{\infty} b(t)h(a)i_0(t,a)\; \d a+\int_0^{\infty} b(t)h(a)i_p(t,a) \; \d a} \right]. 
\end{aligned}
\right.
\end{equation}
As in the model without proliferation, the initial conditions of Eq.~\eqref{Eq:ProliferationMultiScale} are given by
\begin{align*}
V(0) = V_0, \quad  T(0) &= T_0, \quad i_0(0,a) = f_{i_0}(a), \quad  i_p(0,a) = f_{i_p}(a), \\
r_{0}(0,a) &= f_{r_0}(a), \quad  \textrm{and} \quad  r_{p}(0,a) = f_{r_p}(a).
\end{align*}
Once again, we assume that the initial densities $f_{i_0}, \ f_{i_p}, \ f_{r_0},$ and $f_{r_p}$ are non-negative, integrable functions. In the setting of chronic infection, these initial densities can immediately be determined by projecting the characteristics of Eq.~\eqref{Eq:ProliferationMultiScale} backwards in time as described earlier. Finally, within this framework, the total concentration of infected cells is given by 
\begin{align}\label{Eq:TotalInfectedCellsProliferate}
I(t) = \int_0^{\infty}  i_0(t,a)+i_p(t,a) \; \d a,
\end{align}
while the total amount of intracellular vRNA is given by 
\begin{align}\label{Eq:TotalIntracellularvRNAProliferate}
R(t) = \int_0^{\infty}  r_0(t,a)i_0(t,a)+r_p(t,a)i_p(t,a) \;  \d a.
\end{align}
These integrals are well-defined, as both $i_0(t,a)$ and $i_p(t,a)$ decay exponentially in infection age while the intracellular vRNA concentrations are both bounded above. 

\subsection{The multiscale model including infected cell proliferation admits an equivalent system of ODEs}
 
We now demonstrate that the multiscale model Eq.~\eqref{Eq:ProliferationMultiScale} is mathematically equivalent to a system of ODEs under the assumption that the rates $\delta, \mu$, and $\rho$ are all independent of infection age. At this point, we have not made any assumptions on $h(a)$. We begin by deriving an ODE for the total amount of vRNA within infected cells, given by $R(t)$ in Eq.~\eqref{Eq:TotalIntracellularvRNAProliferate}. We make the change of variables $u = t-a$ to find
\begin{align*}
R(t) = \int_{-\infty}^{t} i_0(t,t-u)r_0(t,t-u)\d u + \int_{-\infty}^{t} i_p(t,t-u)r_p(t,t-u)\d u.
\end{align*}
Using Leibniz's rule to differentiate the integral and recalling that the integrands is a product of exponentially decaying and bounded functions which vanishes at infinity \citep{Perthame2006,Cassidy2021}, we obtain 
\begin{align*}
\TimeDeriv R(t) & = i_0(t,0)r_0(t,0) + \int_0^{\infty} r_0(t,a)\left( \partial_t + \partial_a \right)  i_0(t,a) \d a +\int_0^{\infty}i_0(t,a) \left( \partial_t + \partial_a \right) r_0(t,a) \d a \\
&{} \quad + i_p(t,0)r_p(t,0) + \int_0^{\infty} r_p(t,a)\left( \partial_t + \partial_a \right)  i_p(t,a) \d a +\int_0^{\infty}i_p(t,a) \left( \partial_t + \partial_a \right)  r_p(t,a) \d a.
\end{align*}
Using the PDEs for $i_0(t,a),\ i_p(t,a),\ r_0(t,a),$ and $r_p(t,a)$ in Eq.~\eqref{Eq:ProliferationMultiScale}, gives 
\begin{align*}
\TimeDeriv R(t) & = i_0(t,0)r_0(t,0) + \alpha I(t) - (\delta+\mu+\rho)R(t) \\
&{} \ -  \int_0^{\infty} b(t) h(a) \left[ r_0(t,a) i_0(t,a) + r_p(t,a) i_p(t,a) \right]\; \d a + i_p(t,0)r_p(t,0) .
\end{align*}
Then, the boundary conditions in Eq.~\eqref{Eq:BoundaryConditionsProliferationMultiScale} give
\begin{align*}
i_0(t,0)r_0(t,0) & = \zeta \beta V(t) T(t), \\
i_p(t,0)r_p(t,0) & =  \int_0^{\infty}  b(t) h(a) \left[ r_0(t,a)i_0(t,a) + r_p(t,a)i_p(t,a) \right] \d a . 
\end{align*}
We therefore obtain
\begin{align*}
\TimeDeriv R(t) & = \zeta \beta V(t) T(t) + \alpha I(t)  - (\delta + \mu+ \rho)R(t). 
\end{align*}
We note that this differential equation is precisely the same as what was found by \citet{Kitagawa2019} in the multiscale model that does not include proliferation of infected cells. This is unsurprising, as proliferation of infected hepatocytes does not change the total amount of intracellular vRNA since the vRNA within a proliferating mother cell is conserved as it is divided between the two daughter cells.  

We use the same approach for $I(t)$ defined by Eq.~\eqref{Eq:TotalInfectedCellsProliferate} to obtain
 \begin{align*}
\TimeDeriv I(t) & = i_0(t,0)+i_p(t,0) - \delta I(t) - \int_0^{\infty} b(t)h(a)\left[ i_0(t,a)+i_p(t,a) \right] \; \d a .
\end{align*}
Inserting the boundary conditions in Eq.~\eqref{Eq:BoundaryConditionsProliferationMultiScale} gives
\begin{align*} 
\TimeDeriv I(t) & =  \beta V(t) T(t) +  b(t) \int_0^{\infty} h(a)\left[ i_0(t,a)+i_p(t,a) \right] \d a - \delta I(t).
\end{align*}

Now, we have not yet specified $h(a)$. Assuming the likelihood that an infected cell proliferates is independent of infection age, the hazard rate is a constant function of infection age. Thus, $h(a) = \eta >0 $, which implies that the probability of an infected cell proliferating is exponentially distributed. Importantly, if $\eta = 1$, then infected and uninfected hepatocytes are equally likely to proliferate, while taking $\eta \in (0,1)$ would account for a proliferative fitness cost related to HCV infection, as has been observed experimentally \citep{Webster2013}, and $\eta > 1$ would model a proliferative fitness advantage. We stress that other assumptions on $h(a)$ are possible that would also admit equivalent ODEs via the linear chain technique \citep{Cassidy2020a,Diekmann2017,Cassidy2018a}. 

Then, assuming $h(a)= \eta$ and recalling the definition of $I(t)$ in Eq.~\eqref{Eq:TotalInfectedCellsProliferate} gives
\begin{align} \label{Eq:InfectedCellProliferatingODE}
\TimeDeriv I(t) & =  \beta V(t) T(t) +  \eta b(t) I(t) - \delta I(t). 
\end{align}
We thus obtain the equivalent ODE system to the multiscale model with proliferation in Eq.~\eqref{Eq:ProliferationMultiScale}, which is given by
\begin{equation} \label{Eq:ProliferationEquivalentODE}
\left \{ 
\begin{aligned}
\TimeDeriv T(t) & = b(t)T(t)  -\beta V(t) T(t), \\
\TimeDeriv I(t) & = \beta V(t)T(t) +  \eta b(t) I(t)  - \delta I(t), \\
\TimeDeriv R(t) & = \zeta \beta V(t) T(t) +  \alpha I(t)  - (\rho+\mu+\delta)R(t), \\
\TimeDeriv V(t) & =  \rho R(t)  - cV(t).
\end{aligned}
\right. 
\end{equation}
Here, the initial conditions for $I$ and $R$ once again arise naturally from the initial densities of the multiscale model Eq.~\eqref{Eq:ProliferationMultiScale}.
Now, if $b(t)$ is chosen so that hepatocytes undergo logistic growth with a carrying capacity $T_{max}$ and growth rate $\gamma$, as in \citep{Dahari2007,Dahari2007a,Iwanami2020,Reluga2009}, we obtain
\begin{equation} \label{Eq:LogisticProliferationEquivalentODE}
\left \{
\begin{aligned}
\TimeDeriv T(t) & = \gamma \left(1-\frac{T(t)+I(t)}{T_{max}} \right) T(t) -\beta V(t) T(t), \\
\TimeDeriv I(t) & = \beta V(t)T(t) +  \eta   \gamma \left(1-\frac{T(t)+I(t)}{T_{max}} \right)  I(t)  - \delta I(t), \\
\TimeDeriv R(t) & = \zeta \beta V(t) T(t) +  \alpha I(t)  - (\rho+\mu+\delta)R(t), \\
\TimeDeriv V(t) & =  \rho R(t)  - cV(t).
\end{aligned}
\right. 
\end{equation}
We compare Eq.~\eqref{Eq:LogisticProliferationEquivalentODE} against the \citet{Elkaranshawy2024} model in the Appendix and show that the \citet{Elkaranshawy2024} model predicts non-realistic production of intracellular vRNA due to infected cell proliferation.

\subsection{Generation tracking of infected hepatocytes}\label{Sec:GenerationTracking}
In the development of Eq.~\eqref{Eq:ProliferationMultiScale}, we distinguished between cells that arise from new infections, $i_0(t,a)$, and cells that arise from infected cell proliferation, $i_p(t,a)$. In this setting, the boundary condition that defines the amount of intracellular vRNA inherited by newly proliferated daughter cells corresponds to the average amount of vRNA within all proliferating infected cells. However, the amount of vRNA at proliferation of the mother cell determines the amount of intracellular vRNA inherited by the resulting daughter cells. As we expect intracellular vRNA to be produced by infected  cells, and thus accumulate from birth until mitosis, the daughter cells in generation $n$ will have more intracellular vRNA upon mitosis than daughter cells in generation $n-1$. 

However, the resulting generational accumulation of intracellular vRNA is not reflected in the boundary conditions of Eq.~\eqref{Eq:ProliferationMultiScale}, where all daughter cells inherit half the average amount of vRNA of all proliferating mother cells, regardless of their generation. We now show how to account for these generational dynamics via an infinite system of PDEs that tracks the dynamics of each generation of infected hepatocytes. In what follows, the dynamics of newly infected cells, $i_0$, are unchanged from Eq.~\eqref{Eq:ProliferationMultiScale}.

\subsubsection*{Infinite systems of differential equations to model generational inheritance}

We begin by developing a mathematical model for each generation of infected hepatocyte. For $n \geq 1,$ let $i_n(t,a)$ denote the concentration of infected hepatocytes in the $n-$th generation, and let $I(t)$ denote the total concentration of infected hepatocytes. As before, the proliferation rate $b$ may depend on the concentration of uninfected and infected hepatocytes, $T(t)$ and $I(t)$. Then, using the same reasoning as for Eq.~\eqref{Eq:ProliferationMultiScale}, $i_n(t,a)$ satisfies
\begin{align*}
(\partial_t+ \partial_a)i_n(t,a) & =  - [\delta_n +b(t,I(t),T(t))h_n(a)] i_n(t,a).
\end{align*}
In what follows, we assume that the per capita death rate $\delta_n$ and likelihood of proliferation $h_n$ are independent of generation, so $\delta_n = \delta$ and $h_n(a) = h(a)$. For notational simplicity, we suppress the explicit dependence of $b$ on the total concentration of hepatocytes and instead write $ b(t, I(t),T(t)) = b(t).$  The boundary condition of $i_n(t,a)$ corresponds to proliferation of cells in the preceding generation, so the production of infected hepatocytes in generation $n$ is given by
\begin{align*}
i_n(t,0) = 2 \int_{0}^{\infty} b(t)h(a)i_{n-1}(t,a)\d a.
\end{align*}
In this framework, $a$ represents the time-since infection for $i_0(t,a)$ and the time since mitosis for $i_n(t,a)$ with $n = 1,2,...$.
Similarly, let $r_n(t,a)$ denote the amount of intracellular vRNA within infected cells of age $a$ in generation $n$. Once again assuming that the dynamics of intracelluar vRNA are independent of generation and using the same reasoning as in Eq.~\eqref{Eq:ProliferationMultiScale}, we obtain
\begin{align*}
(\partial_t+ \partial_a)r_n(t,a) & = \alpha  - [\mu+\rho] r_n(t,a).
\end{align*}
Then, the boundary condition of $r_n(t,a)$ corresponds to the amount of vRNA inherited by infected cells in generation $n$ following mitosis of the mother cell in generation $n-1$. The total amount of vRNA within these proliferating cells is given by 
\begin{align*}
 \int_{0}^{\infty} b(t)h(a)i_{n-1}(t,a)r_{n-1}(t,a)\; \d a,
\end{align*}
and this vRNA is equally divided among the resulting daughter cells, now in the $n$-th generation, which determines the boundary condition $i_n(t,0)$. Consequently, each infected cell in generation $n$ inherits 
\begin{align*}
r_n(t,0) = \frac{\int_{0}^{\infty} b(t)h(a)i_{n-1}(t,a)r_{n-1}(t,a) \; \d a }{2 \int_{0}^{\infty} b(t)h(a)i_{n-1}(t,a) \; \d a},
\end{align*}
copies of vRNA. As we are not \textit{a priori} imposing an upper limit on the number of generations, $n$, the resulting model is a coupled infinite system of PDEs for the generational dynamics of $i_n(t,a)$ and $r_n(t,a)$ with $n = 1,2,3,...$. As in Section~\ref{Sec:MultiscaleProliferation}, the total concentration of infected cells in the $n-$th generation is given by  
\begin{align*}
I_n(t) = \int_0^{\infty} i_n(t,a) \d a.
\end{align*}
 Then, using the same approach as for Eq.~\eqref{Eq:ProliferationEquivalentODE}, we calculate
 \begin{align*}
 \TimeDeriv I_n(t) = 2 b(t) \int_0^{\infty} h(a) i_{n-1}(t,a)\d a - \delta I_n(t) - b(t) \int_0^{\infty} h(a) i_{n}(t,a)\d a .
 \end{align*}
Under the assumption that $h(a) = \eta$, this differential equation for $I_n(t)$ becomes
\begin{align*}
 \TimeDeriv I_n(t) = 2 b(t)\eta I_{n-1}(t)  - \delta I_n(t) - b(t)\eta I_n(t),
 \end{align*}
where we clearly see the distinction between the loss of cells in generation $n$ due to mitosis and the appearance of cells in this generation due to proliferation of cells in generation $n-1$. In this way, we obtain an infinite system of ODEs for $\{ I_n(t) \}_{n=1}^{\infty}$. Assuming that $I_m(0) \geq 0$ for all $m$, it is simple to show by induction and using Gronwall's lemma that $I_m(t) \geq 0 $ for all $m$ \citep{Reid1930}. Returning to the intracellular vRNA concentration, the total amount of vRNA within each generation is given by
\begin{align*}
R_n(t) = \int_0^{\infty} i_n(t,a)r_n(t,a) \d a,
\end{align*}
and using the same approach as in Section~\ref{Sec:MultiscaleProliferation}, we calculate
\begin{align*}
\TimeDeriv R_n(t) = \alpha I_n(t)  - (\delta + \mu+ \rho)R_n(t) .
\end{align*}
As with the concentration of infected hepatocytes in the $n$-th generation, assuming that $R_m(0) \geq 0$ for all $m$ implies that $R_m(t) \geq 0$. We thus obtain the infinite system of ODEs
\begin{equation}\label{Eq:InfiniteSystemODE}
\left \{
\begin{aligned}
\TimeDeriv T(t) & =b(t) T(t) -\beta T(t) V(t), \\
\TimeDeriv I_0(t) &  =  \beta T(t) V(t)  - \eta b(t) I_0(t)-\delta I_0(t), \\
\TimeDeriv I_n(t) & =    2\eta b(t) I_{n-1}(t) -  \eta b(t) I_n(t)  - \delta I_{n}(t) , \  \textrm{for} \  n \geq 1,  \\
\TimeDeriv R_0(t) &  = \zeta \beta T(t) V(t) + \alpha I_0(t)  - (\delta + \mu+ \rho)R_0(t), \\
\TimeDeriv R_n(t) & = \alpha I_n(t)  - (\delta + \mu+ \rho)R_n(t), \  \textrm{for} \  n \geq 1,  \\
\TimeDeriv V(t) & = \rho \sum_{n=0}^{\infty} R_n(t) - cV(t).
\end{aligned}
\right.
\end{equation}

\subsubsection*{Tracking generational dynamics does not influence infection dynamics}

We now show that the dynamics of the infinite system of ODEs in Eq.~\eqref{Eq:InfiniteSystemODE} are precisely the same as those of the multiscale model with proliferation that did not track the generational inheritance. To do so, we remark that, if we could interchange infinite summation and differentiation, so that
\begin{align*}
\TimeDeriv  \displaystyle \sum_{n=0}^{\infty} I_n(t) =  \displaystyle  \sum_{n=0}^{\infty} \TimeDeriv I_n(t),
\end{align*}
then the resulting infinite series would telescope to exactly the ODE for $I(t)$ in Eq.~\eqref{Eq:LogisticProliferationEquivalentODE}. In Appendix~\ref{Sec:InfiniteODEProof}, we utilize the results on infinite systems of differential equations from \citet{McClure1976}, where they considered the generic infinite system of ODEs
\begin{align}\label{Eq:McClureODE}
\TimeDeriv y_i(t) = \displaystyle \sum_{j=1}^{\infty} a_{i,j}(t)y_j(t) + f_j(\tilde{y}(t)), \quad \textrm{for} \quad i = 1,2,...,
\end{align}
with $\tilde{y} = (y_1(t), y_2(t),...)$, to show that we can exchange differentiation and summation. \citet{McClure1976} demonstrated that a unique solution of Eq.~\eqref{Eq:McClureODE} exists and is a \textit{strongly continuous function}.  
\sloppy
\begin{definition} [Strongly continuous function]
The sequence $\tilde{y}(t) = (y_1(t), y_2(t),...)$ is a strongly continuous function on the domain $\mathcal{I} \subset \mathbb{R}$ if each $y_i(t)$ is continuous and
\begin{align*}
\|\tilde{y}(t) \| = \displaystyle \sum_{i=1}^{\infty} |y_i(t)|
\end{align*}
is uniformly convergent on compact subsets of $\mathcal{I}$.
\end{definition}

In Appendix~\ref{Sec:InfiniteODEProof}, we prove 
\begin{proposition}\label{Prop:StronglyContinuousSolution}
Assume that the parameters in Eq.~\eqref{Eq:InfiniteSystemODE} are strictly positive, the proliferation function $b(t)$ is continuous and bounded, and that the initial conditions are non-negative. Then, the infinite system of ODEs~Eq.~\eqref{Eq:InfiniteSystemODE} has a strongly continuous solution. 
\end{proposition}

Now, let $\tilde{I}(t) = (I_1(t), I_2(t), ...)$ and $\tilde{R}(t) = (R_1(t), R_2(t),...),$ be the strongly continuous solution of Eq.~\eqref{Eq:InfiniteSystemODE} and define
\begin{align*}
J_N(t) = I_0(t) + \displaystyle \sum_{n=1}^N I_n(t) \quad \textrm{and} \quad P_N(t) = R_0(t) + \displaystyle \sum_{n=1}^N R_n(t).
\end{align*}
Proposition~\ref{Prop:StronglyContinuousSolution} ensures that $J_N$ and $P_N$ converge uniformly, so it is natural to define the total concentration of infected cells, $I(t)$, and intracellular vRNA, $R(t)$, as
\begin{align*}
I(t) = \lim_{N \to \infty} J_N(t)  \quad \textrm{and} \quad R(t) =  \lim_{N \to \infty} P_N(t). 
\end{align*}
Then, by adding and subtracting the term $\eta b(t) I_N(t)$ in the ODE for $J_N$, we calculate
\begin{align*}
\TimeDeriv J_N(t) & = \beta V(t)T(t) +   \eta b(t) J_{N}(t)  - \delta  J_{N}(t)  - 2\eta b(t) I_{N}(t), \\
\TimeDeriv P_N(t) & = \zeta \beta V(t) T(t) +  \alpha J_N(t)  - (\delta + \mu+ \rho)P_N(t).
\end{align*}
Since $J_N$ and $P_N$ converge uniformly and $I_N\to 0$ as $N \to \infty$, the sequences $J'_N$ and $P'_N$ converge uniformly, so after interchanging differentiation and summation, we find 
\begin{align*}
\TimeDeriv I(t)  & = \beta V(t)T(t) +  \eta b(t) I(t)  - \delta I(t), \ \textrm{and}  \\
\TimeDeriv R(t)  & = \zeta \beta V(t) T(t) +  \alpha I(t)  - (\rho+\mu+\delta)R(t), 
\end{align*}   
 which are precisely what we obtained in Eq.~\eqref{Eq:ProliferationEquivalentODE}.  Consequently, the dynamics of the total concentration of infected hepatocytes and intracellular vRNA in the infinite ODE that tracks generation inheritance Eq.~\eqref{Eq:InfiniteSystemODE} are exactly those of Eq.~\eqref{Eq:ProliferationEquivalentODE}, where we only considered infected hepatocytes arising from new infections or from proliferation of infected mother cells. We therefore do not need to account for the generational dynamics of intracellular vRNA, as the inheritance balances over all generations. Further, for non-negative parameters, the existence, uniqueness, and non-negativity of solutions to Eq.~\eqref{Eq:InfiniteSystemODE} is immediate.
 
 
\section{Bifurcation analysis of the multiscale model with logistic proliferation}

 While other choices for the proliferation function $b$ are possible, the majority of models that consider hepatocyte proliferation utilize a logistic growth term \citep{Dahari2007,Dahari2009a,Reluga2009}. We therefore now analyse the bifurcation structure the multiscale model with logistic proliferation in Eq.~\eqref{Eq:LogisticProliferationEquivalentODE}. We give necessary and sufficient conditions for the existence and local stability of the uninfected equilibrium, a total infection equilibrium, where the infected cell population is self-sustaining in the absence of uninfected hepatocytes, and an infected equilibrium. We illustrate these results using Matcont \citep{Dhooge2008} in Fig.~\ref{Fig:NumericalBifurcationResults}.

\subsection*{Uninfected equilibrium}
We begin by establishing the local stability of the uninfected equilibrium. As would be expected from the biological interpretation of Eq.~\eqref{Eq:LogisticProliferationEquivalentODE}, the uninfected equilibria is given by 
\begin{align*}
(T^*,I^*,R^*,V^*) = (T_{max},0,0,0).
\end{align*}
The Jacobian of Eq.~\eqref{Eq:LogisticProliferationEquivalentODE} evaluated at the uninfected equilibrium is given by 
\begin{equation}
J(T_{max},0,0,0) = \left[ 
\begin{array}{cccc}
-\gamma & -\gamma & 0 & -\beta T_{max} \\
 0 & -\delta & 0 &  \beta T_{max} \\
 0 & \alpha & -(\rho+\mu+\delta) & \beta T_{max} \\
 0 & 0 & \rho & -c \\
\end{array}
\right].
\end{equation}
\sloppy 
The Jacobian matrix $J(T_{max},0,0,0)$ has an eigenvalue $\lambda_0 = -\gamma$, so after expanding $J(T_{max},0,0,0)$ along the first column, the remaining three eigenvalues are the roots of 
\begin{align*}
CP_M(\lambda) = [\beta \rho T_{max} - (\rho +\mu+\delta + \lambda)(c+\lambda)](\delta + \lambda) + \beta T_{max} \alpha \rho. 
\end{align*}
This is precisely the characteristic polynomial found by \citet{Elkaranshawy2024} in their analysis of the uninfected equilibrium in their similar model of HCV infection under the assumption that $T^* = T_{max}$. By utilizing the Routh-Hurwitz conditions and adapting the analysis of \citet{Elkaranshawy2024}, we conclude that the uninfected equilibria is locally asymptotically stable if and only if the basic reproductive number $\mathcal{R}$ satisifies 
\begin{align}\label{Eq:ReproductionNumber}
\mathcal{R} = \frac{\beta T_{max}(\alpha + \delta) \rho}{\delta c (\delta+\mu+\rho) } < 1.
\end{align}

\subsection*{Total infection equilibria of the multiscale model}

We now turn to persistently infected equilibrium of Eq.~\eqref{Eq:LogisticProliferationEquivalentODE} with non-zero equilibrium concentrations of infected hepatocytes $I^*$. From the differential equation for target cells in Eq.~\eqref{Eq:LogisticProliferationEquivalentODE}, $T(t)$, we find
\begin{align}\label{Eq:InfectedEquilibriumInfectionRate}
\beta T^* V^* = \gamma T^* \left(1-\frac{T^*+I^*}{T_{max}} \right). 
\end{align}
Now, we typically discard the equilibrium solution corresponding to $ T^{*} =0$. However, when we consider proliferation of infected hepatocytes, there may be an infected equilibrium corresponding to a population of infected hepatocyte that is self-sustaining \textit{without} new infections. We denote this total infection equilibrium by $( 0 ,I^{\dagger},R^{\dagger}, V^{\dagger})$. At this equilibrium, we immediately obtain  from Eq.~\eqref{Eq:LogisticProliferationEquivalentODE}
\begin{align*}
 \eta  \gamma \left[ \left(1 - \frac{\delta}{\eta \gamma} \right) - \frac{I^{\dagger}}{T_{max}} \right]I^{\dagger} = 0, 
\end{align*}
with corresponding equilibrium concentration $I^{\dagger} = ( 1-\frac{\delta}{\eta \gamma} )T_{max}$. This equilibrium is only biologically relevant if $ \eta  \gamma > \delta, $ and we calculate the remaining equilibrium concentrations
\begin{align*}
R^{\dagger} = \frac{\alpha( \eta  \gamma - \delta  )T_{max} }{\eta  \gamma (\rho+\mu+\delta)} \quad \textrm{and} \quad V^{\dagger} =  \frac{\rho \alpha( \eta  \gamma - \delta  )T_{max} }{ \eta  \gamma c (\rho+\mu+\delta)}.
\end{align*} 

The Jacobian of Eq.~\eqref{Eq:LogisticProliferationEquivalentODE} evaluated at $(0,I^{\dagger},R^{\dagger}, V^{\dagger})$ is given by 
\begin{equation*}
J(0,I^{\dagger},R^{\dagger},V^{\dagger}) = \left[ 
\begin{array}{cccc}
 \frac{\delta}{\eta}  - \beta V^{\dagger} & 0 & 0 & 0 \\
 \beta V^{\dagger} + \eta \gamma - \delta   & -\eta \gamma  & 0 & 0  \\
 \zeta \beta V^{\dagger}  & \alpha & -(\rho+\mu+\delta) & 0 \\
 0 & 0 & \rho & -c \\
\end{array}
\right].
\end{equation*}
This matrix is lower triangular with eigenvalues
 \begin{align*}
 \lambda_1 =  \frac{\delta}{\eta}  - \beta V^{\dagger}, \quad \lambda_2 = -\eta \gamma , \quad  \lambda_3 = -(\rho+\mu+\delta), \quad  \lambda_4 = -c. 
 \end{align*}
It is clear that the final three eigenvalues are strictly negative, so the total infection equilibrium $(0,I^{\dagger},R^{\dagger},V^{\dagger})$ is locally asymptotically stable if and only if $\lambda_1 < 0$. This condition is equivalent to $  \mathcal{R} > \mathcal{R}^{\dagger} $, where $\mathcal{R}$ is the basic reproduction number in Eq.~\eqref{Eq:ReproductionNumber}, and
 \begin{align}\label{Eq:RDaggerDefn}
 \mathcal{R}^{\dagger} = \frac{1}{\gamma \eta} \left( \gamma +  \frac{\rho \beta T_{max} }{  c (\rho+\mu+\delta)} \left( \alpha +  \gamma \eta \right)  \right). 
 \end{align}
 
\subsection*{Infected equilibria of the multiscale model with logistic proliferation}

Now, we consider equilibria of Eq.~\eqref{Eq:ProliferationEquivalentODE} with biologically feasible uninfected hepatocyte concentrations, $T^* \in (0, T_{max})$. The differential equation for circulating virus implies that, at equilibrium, $c V^* = \rho R^*$. Then, using Eq.~\eqref{Eq:InfectedEquilibriumInfectionRate} and this equality in the ODE for the total amount of intracellular vRNA, $R(t)$, in Eq.~\eqref{Eq:ProliferationEquivalentODE} implies
\begin{align*}
\zeta \gamma T^* \left(1-\frac{T^*+I^*}{T_{max}} \right) - (\rho + \mu +\delta) \left( \frac{c}{\rho}\right) \left(\frac{\gamma}{\beta}\right) \left(1-\frac{T^*+I^*}{T_{max}} \right) + \alpha I^* = 0. 
\end{align*}
We note that the above expression is linear in $I^*$, so regrouping terms gives 
\begin{align*}
I^* \left( \frac{r}{T_{max}} \left[ \frac{\zeta \rho \beta T^* - (\rho+\mu+\delta)c}{\rho \beta} 
\right] -\alpha \right) = \gamma  \left(1-\frac{T^*}{T_{max}} \right)\left[ \frac{\zeta \rho \beta T^* - (\rho+\mu+\delta)c}{\rho \beta} \right].
\end{align*}
Now, setting 
\begin{align*}
M(T^*) = \zeta \rho \beta T^* - (\rho + \mu + \delta)c \quad \textrm{and} \quad B = \frac{\alpha \rho \beta T_{max}}{\gamma },
\end{align*}
we find
\begin{align}\label{Eq:InfectedEquilibriumI1}
\frac{I^*}{T_{max}} =  \frac{(1-T^*/T_{max})M(T^*)}{M(T^*)-B}.
\end{align}
Now, we turn to the ODE for infected hepatocytes, and once again using Eq.~\eqref{Eq:InfectedEquilibriumInfectionRate}, we find
\begin{align*}
\delta I^* = \gamma \left(1 - \frac{T^* + I^*}{T_{max}} \right) (T^*+\eta I^*).
\end{align*} 
The equality in Eq.~\eqref{Eq:InfectedEquilibriumI1} yields
\begin{align*} 
\left( 1-\frac{T^*}{T_{max} } \right)  \left( \frac{F(T^*)}{(M(T^*)-B)^2} \right)  = 0,
\end{align*}
which clearly has a solution corresponding to the uninfected equilibrium $T^* = T_{max}$ and where 
\begin{align*}
F(T^*) = \gamma B T^* ( M(T^*)-B) + \gamma B\eta T_{max} (1-T^*/ T_{max}) M(T^*) + \delta T_{max}M(T^*)(M(T^*)-B), 
\end{align*}
is a quadratic function of $T^*$. Then, the uninfected hepatocyte concentration, $T^*$, at an infected equilibrium must satisfy $F(T^*) = 0,$ while
\begin{align} \label{Eq:TMaxTranscritical}
F(T_{max}) = \frac{T_{max}}{\delta c(\rho + \mu +\delta)} \left[ \mathcal{R}-1\right]\left(\rho \beta T_{max}(\zeta-\alpha) - (\rho+\mu+\delta)c \right).
\end{align}
We immediately see that $\mathcal{R} =1$ implies $F(T_{max}) = 0 $ so the uninfected equilibrium undergoes a transcritical bifurcation at $ \mathcal{R}  = 1$ leading to the existence of a positive infected equilibrium. In Appendix~\ref{Appendix:BifurcationOfInfectedEquilibria}, we show that this infected equilibrium coalesces with the total infection equilibrium in a transcritical bifurcation when $\mathcal{R}  = \mathcal{R}^{\dagger}.$ Furthermore, we show that the two roots of $F$ collide in a saddle-node bifurcation if $\mathcal{R}  = \mathcal{R}^{\dagger} $ and  
\begin{align}\label{Eq:EtaSaddleNode}
\eta = 1  + \frac{\delta\zeta}{\alpha} + \frac{ \alpha \beta \rho T_{max}	}{\gamma(\rho+\mu+\delta)c} .
\end{align}
The first condition, $\mathcal{R}=\mathcal{R}^{\dagger}$, is precisely the necessary condition for the total infection equilibrium to undergo a transcritical bifurcation. Indeed, if both Eq.~\eqref{Eq:EtaSaddleNode} and $\mathcal{R}=\mathcal{R}^{\dagger}$ hold, then there is a saddle-node transcritical (SNTC) bifurcation of the total-infection equilibrium \citep{Donohue2020}. This SNTC bifurcation is a co-dimension two bifurcation where a saddle-node bifurcation collides with a transcritical bifurcation. These bifurcations have been observed in ecological models with highly nonlinear terms \citep{Touboul2018}, but, to our knowledge, not in a system with only quadratic nonlinearities. We illustrate this codimension-two bifurcation in Fig.~\ref{Fig:NumericalBifurcationResults}C-E.


\subsubsection*{Infection induced fitness cost}

We next consider the case where HCV infection results in a proliferative fitness cost for infected hepatocytes, as has been observed experimentally, which corresponds to $\eta \leq 1$ \citep{Webster2013}. In Appendix~\ref{Appendix:InfectedEquilibria}, we characterise the existence of the biologically relevant infected equilibrium under this assumption and prove
\begin{lemma}\label{Lemma:InfectedEquilibrium}
Let the model parameters be positive. Further, assume that 
\begin{align*}
\eta \in [0,1], \quad  \zeta \leq 1, \quad \textrm{and} \quad \alpha > \zeta \max \left[ 1 ,  1-\eta +\sqrt{(1-\eta)^2 + 4\delta} \right].  
\end{align*}
Eq.~\eqref{Eq:LogisticProliferationEquivalentODE} has an uninfected equilibrium $(T_{max},0,0,0)$ and the totally infected equilibrium $(0,I^{\dagger},R, V^{\dagger})$. Further, Eq.~\eqref{Eq:LogisticProliferationEquivalentODE} has a unique infected equilibrium if and only if $1< \mathcal{R} < \mathcal{R}^{\dagger}$.
\end{lemma}

\subsubsection*{Infection induced fitness advantage}
We next turn to the case where HCV infection imparts a proliferative fitness advantage, denoted by $\eta > 1,$ which may occur if HCV infection leads to hepatocellular carcinoma. Typical estimates for $T_{max}, \beta$, and $c$ from fitting clinical data satisfy $\beta T_{max} < c$ \citep{Rong2013,Rong2010}. It follows that
\begin{align*}
1< \frac{1+ \frac{\alpha \rho \beta T_{max} }{  \gamma c (\rho+\mu+\delta)} }{1- \frac{\rho \beta T_{max} }{  c (\rho+\mu+\delta)} } < \eta  \quad \textrm{implies} \quad \mathcal{R}^{\dagger} < 1.
\end{align*} 
 Consequently, the local stability analysis suggests that there is a possible region of bistability of the uninfected and the total infection equilibrium for $\mathcal{R} \in (\mathcal{R}^{\dagger},1)$. This is only possible if the uninfected equilibrium undergoes a backwards bifurcation at $\mathcal{R} =1$. 

This backwards bifurcation is not typical in viral dynamics models but was observed by \citet{Reluga2009} in a variant of the standard model of viral dynamics that includes proliferation of HCV infected hepatocytes. There, \citet{Reluga2009} noted that the viral dynamics are relatively fast compared to the dynamics of infected hepatocytes. Thus, they assumed that these two concentrations were in quasi-steady state so the standard viral dynamics model reduces to a system of two differential equations. Then, \citet{Reluga2009} show that bistability between the uninfected and total infected equilibria can only occur if infected hepatocytes have a proliferative fitness advantage compared to uninfected hepatocytes, which precisely corresponds to the condition $\eta > 1$. We illustrate this backwards bifurcation in Fig.~\ref{Fig:NumericalBifurcationResults}.
 
\subsection*{Summary and numerical bifurcation analysis}

We utilize Matcont, a numerical continuation software package commonly used in mathematical biology \citep{Dhooge2008,Sanche2021}, to illustrate these analytical results. We take $\mathcal{R}$ as the bifurcation parameter and calculate the corresponding infection rate $\beta$ by re-writing Eq.~\eqref{Eq:ReproductionNumber}  
\begin{align*}
\beta = \mathcal{R} \frac{\delta c (\delta+\mu+\rho)}{ T_{max}(\alpha + \delta) \rho}.
\end{align*}
We use previously estimated model parameters \citep{Elkaranshawy2024,Rong2013a,Kitagawa2018,Rong2013} given in Table~\ref{Table:ModelSimulationParameters} in this numerical bifurcation analysis. However, we consider logistic growth, rather than constant production, of uninfected hepatocytes. Consequently, there are fewer uninfected hepatocytes produced without the constant production rate $\lambda$. Therefore, more virus must be produced by each infected cell to allow for persistent infection. Thus, we increase the production rate of intracellular vRNA over previously published estimates \citep{Elkaranshawy2024,Rong2013a,Kitagawa2018} to  $\alpha = 300$ virions/cell/day to ensure that the predicted dynamics of $V(t)$ are consistent with clinically observed viral loads. Finally, we set the hepatocyte reproduction rate $\gamma = 0.2$/day throughout, although our qualitative results are robust to differences in these parameters.

\begin{table}[!ht]
\centering
\setlength\tabcolsep{1.5pt}
\noindent 
\begin{tabular*}{0.95\hsize}{@{\extracolsep{\fill}}  cll}\sphline 
Parameter & Value   \\ \hline
$\rho$ (/day) & 8.18  \\
$\mu$ (/day) & 1  \\
$c$ (/day) & 22.3    \\
$T_{max}$ (cells/mL) & 1.35$\times 10^7$   \\
$\zeta$ (copies/cell) & 1    \\
$ \gamma $ (/day) &  0.20  \\ 
$\alpha$ (copies/cell/day) &  300   \\ 
$\delta$ (/day) &  0.05   \\ \hline   
\end{tabular*}
\caption{\textbf{The viral dynamics parameter values utilized in the numerical bifurcation analysis.} The model parameters were taken from previously published estimates \citep{Elkaranshawy2024,Rong2013a,Kitagawa2018,Rong2013} except for $\alpha$, which was chosen to ensure that the predicted viral loads are consistent with clinically observed data. Further, $\eta$ was varied to illustrate the possible bifurcation diagrams. As we assumed logistic proliferation of uninfected hepatocytes, we did not consider their proliferation or per capita death rate and set $\lambda = 0$ and $d_T = 0$ throughout.}
\label{Table:ModelSimulationParameters}
\end{table}

\begin{figure}[htbp!] 
\centering  
\begin{tabular}{c} 
 \includegraphics[trim= 13 10 9 5,clip,width=0.90\textwidth]{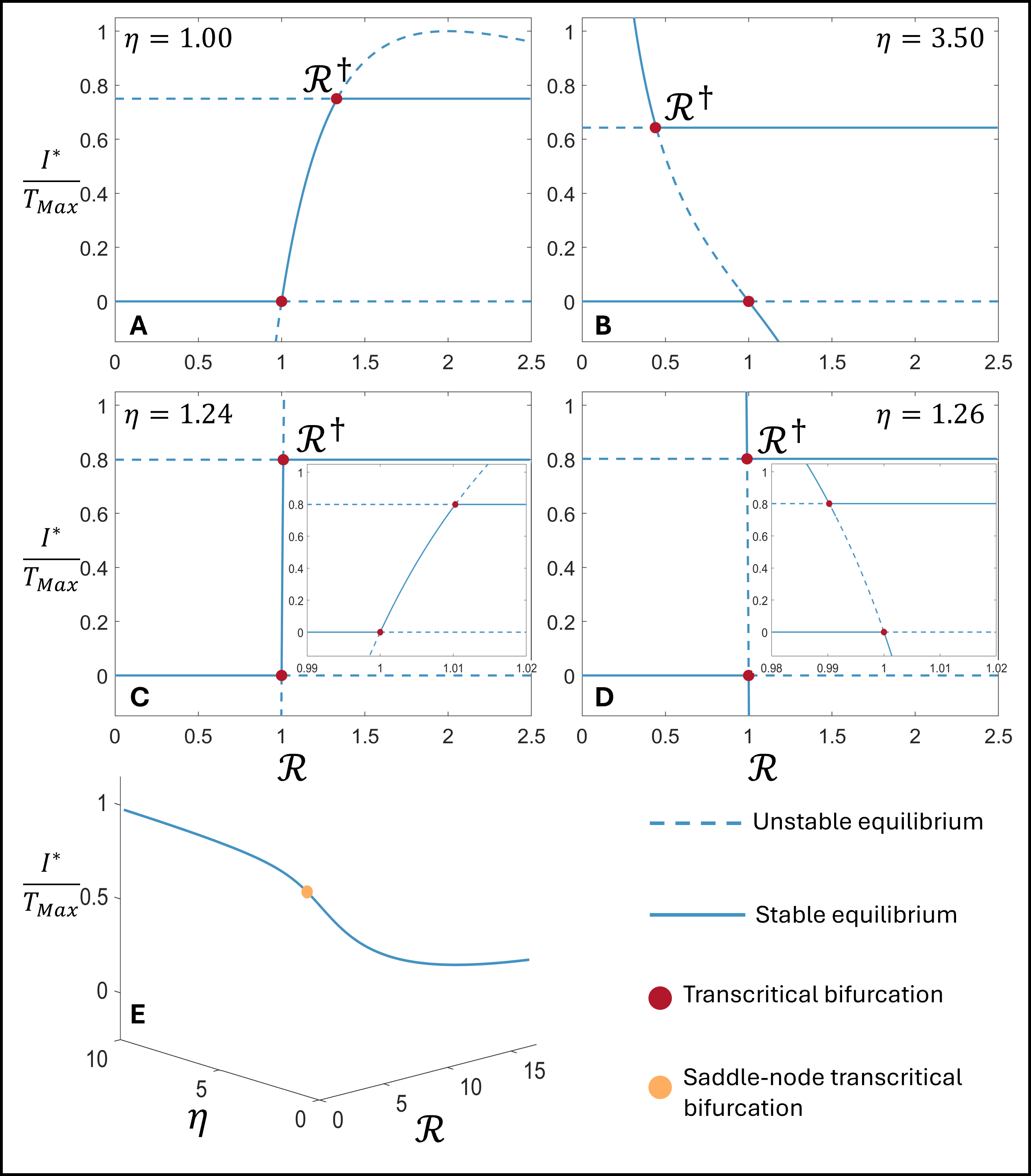} 
\end{tabular}
\caption{ \textbf{Numerical bifurcation analysis of Eq.~\eqref{Eq:LogisticProliferationEquivalentODE}} All panels show the fraction of infected hepatocytes at equilibria, $I^*/T_{max}$ with parameters given in Table~\ref{Table:ModelSimulationParameters}. In Panels A-D, solid lines denote locally stable equilibria, dashed lines denote unstable equilibria, and transcritical bifurcations are denoted by red circles. \textbf{(A)} For $\eta = 1$, the uninfected equilibrium undergoes a forward transcritical bifurcation at $\mathcal{R} = 1$. The resulting infected equilibrium coalesces with the total-infection equilibrium and undergoes a second transcritical bifurcation at $\mathcal{R} = \mathcal{R}^{\dagger}$. \textbf{B)} For $\eta = 3.5$ and $\mathcal{R}^{\dagger} < 1$, the infection-free equilibrium undergoes a backwards bifurcation at $\mathcal{R}=1$, leading to the existence of an unstable infected equilibrium, and bistability between the uninfected and the total-infection equilibrium for $\mathcal{R}^{\dagger} < \mathcal{R} < 1$. \textbf{C)} Unfolding the SNTC bifurcation for $\eta = 1.26$, where the uninfected equilibrium undergoes a backwards transcritical bifurcation at $\mathcal{R} = 1$. \textbf{D)} Unfolding the SNTC bifurcation for $\eta = 1.26$, where the uninfected equilibrium undergoes a backwards transcritical bifurcation at $\mathcal{R} = 1$. In both \textbf{C} and \textbf{D}, the inset diagrams show the forwards and backwards bifurcations, respectively. \textbf{E)} Two parameter continuation of the transcritical bifurcation of the total-infection equilibrium as a function of the model parameters $\mathcal{R}$ and $\eta$. The solid line denotes the value of $I^{*}/T_{max}$ as both parameters vary. The saddle-node transcritical bifurcation is denoted by the yellow circle at $\eta = 1.2501$ and $\mathcal{R} = 1.0000.$ }
\label{Fig:NumericalBifurcationResults}
\end{figure}

We have shown that the reproduction number $\mathcal{R}$ in Eq.~\eqref{Eq:ReproductionNumber} determines the local stability of equilibrium solutions of Eq.~\eqref{Eq:LogisticProliferationEquivalentODE}. In Fig.~\ref{Fig:NumericalBifurcationResults}A, we show the bifurcation diagram for $\eta=1$, where $\mathcal{R}^{\dagger} >1$. Here, Eq.~\eqref{Eq:LogisticProliferationEquivalentODE} has the forward bifurcation structure typically found in viral dynamics models. Specifically, we find that if $\mathcal{R} < 1$, then the uninfected equilibrium is locally stable. There is a transcritical bifurcation of the uninfected equilibrium at $\mathcal{R} = 1$. Then, the resulting infected equilibrium is locally stable for $1 < \mathcal{R} < \mathcal{R}^{\dagger}$. Finally, this infected equilibrium undergoes a further transcritical bifurcation at $\mathcal{R} = \mathcal{R}^{\dagger}$ where the total infection equilibrium becomes locally stable. 

In Fig.~\ref{Fig:NumericalBifurcationResults}B, we consider the case where $\eta = 3.5$ and $\mathcal{R}^{\dagger} < 1$. Unlike the standard viral dynamics model, we observe a backwards bifurcation at $\mathcal{R} = 1$, where the unstable unique infected equilibrium acts as a separatrix between the uninfected and the total infection equilibria, which are both stable. There is consequently a regime of bistability in Eq.~\eqref{Eq:LogisticProliferationEquivalentODE} that corresponds to a population of infected hepatocytes that is self-sustaining due to proliferation. This total infected equilibrium is then locally stable for $ \mathcal{R} > \mathcal{R}^{\dagger} .$ 

Finally, we utilize two parameter continuation of the model parameters $\mathcal{R}$ and $\eta$ to illustrate the unfolding of the SNTC bifurcation of the total-infection equilibrium. The total infection equilibrium is biologically relevant for $\eta > \delta/r = 0.25$ and the SNTC bifurcation occurs at $\eta = 1.2501$ and $\mathcal{R} = 1.0000$. In Fig.~\ref{Fig:NumericalBifurcationResults}C and D, we show an unfolding of the SNTC bifurcation with $\eta = 1.26$ and $\eta = 1.24$. Here, we see the near vertical infection equilibrium following the transcritical bifurcation of the uninfected equilibrium, with a backwards bifurcation occurring at $\eta = 1.26$ and a forward bifurcation occurring at $\eta = 1.24$; these curves coalesce into a cusp at the SNTC bifurcation point. Further, Fig.~\ref{Fig:NumericalBifurcationResults}E shows the continuation of the transcritical bifurcation between the infected  and the total infection equilibria. We note that, as infected cells have an increased fitness advantage with $\eta >1$, the fraction $I^{\dagger}/T_{max}$ approaches unity despite the expected number of secondary infections, $\mathcal{R}$, tending to 0.

\section{Discussion}
The standard multiscale model of HCV infection is a coupled system of integro-partial differential equations that has been instrumental in our understanding of HCV viral dynamics. While the standard multiscale model includes a detailed description of the intracellular HCV life cycle, it does not include infected cell proliferation. However, infected cell proliferation, which may occur at a different rate than for uninfected hepatocytes, may be an important driver of HCV persistence. There is a need for modeling approaches that capture both the intracellular viral life cycle and proliferation of infected hepatocytes. However, including cell proliferation within these coupled PDE and integrodifferential equations is complicated by the inheritance of intracellular viral material to both daughter cells. Appropriately accounting for this generational inheritance is non-trivial in structured models \citep{Metz1986} and it is tempting to immediately adapt equivalent ODE formulations of the multiscale model to include the proliferation of infected hepatocytes without first considering the underlying structured PDE. As we showed in Appendix~\ref{Appendix:Comparison}, direct adaptations of the ODE model without considering the underlying multiscale PDE may lead to non-biological production of intracellular vRNA during proliferation.  

Here, we developed a multiscale model of HCV infection that includes proliferation of infected hepatocytes and the corresponding inheritance of intracellular virus. We included proliferation of infected hepatocytes through the boundary conditions of the time-since-infection structured PDE. In our modeling framework, we distinguished between newly infected hepatocytes and infected hepatocytes arising via proliferation of infected cells. Consequently, the boundary conditions of our structured PDE accurately represent the distinct mechanisms by which infected cells with age $a = 0$ may arise. This distinction is crucial when modeling the inheritance of intracellular viral material, as the initial concentration of vRNA can be drastically different between newly infected cells and those that arise from proliferation. After developing our  multiscale model of inheritance, we show that the model is mathematically equivalent to a system of ODEs under the assumption that the rate constants for intracellular production and degradation of vRNA and the per capita death rate of infected cells are independent of infection age. We also show that this modeling approach that only differentiates between newly infected hepatocytes and those that arise from proliferation is precisely equivalent to a model that explicitly tracks the generational dynamics of infected hepatocytes using techniques from infinite systems of differential equations.

We calculated the basic reproduction number and performed a detailed bifurcation analysis of the resulting model using both analytical and numerical continuation methods. We identified a region of bistability where both the uninfected and total-infection equilibrium are stable. This region of bistability results from the proliferation of infected hepatocytes occurring more rapidly than proliferation of uninfected cells, which might occur in cases where HCV infection leads to hepatocellular carcinoma. If the fitness advantage of infected cells is large enough, then the bistability would persist even in the case of extremely effective antiviral treatment that drives the expected number of secondary infections, $\mathcal{R}$, to zero. Furthermore, while Eq.~\eqref{Eq:LogisticProliferationEquivalentODE} is relatively simple, with only quadratic nonlinearities, we also demonstrated the existence of a saddle-node transcritical bifurcation. This bifurcation separates parameter regimes where the uninfected equilibrium undergoes a forward or backward transcritical bifurcation. While these saddle-node transcritical bifurcations have been observed in highly nonlinear ecological models, to our knowledge, this is the first example in a viral dynamics model with only quadratic nonlinearities.

It is important to note that our analysis relies on the assumption that the rate constants describing intracellular vRNA production, degradation, and the per capita death rate of infected hepatocytes are constant. However, these rate constants will likely depend on the cell's infected age, and possibly, on their infected generation. Consequently, future work may include extending our modeling framework to the more biologically relevant case of age dependent parameters as in \cite{Nelson2004,Hailegiorgis2023}. In recent work along these lines, \citet{Wang2024} analysed a similar multiscale model without proliferation and demonstrated the existence of a forward bifurcation of the uninfected equilibrium without assuming that the production rate of intracellular vRNA is constant.  

All told, we have developed a multiscale modeling framework that captures both the intracellular dynamics of HCV infection as well as cellular proliferation and inheritance. This modeling framework could be extended to other persistent viral infections or to understand inheritance of intracellulars factors in other therapeutic settings.

\section*{Acknowledgements}

Portions of this work were done under the auspices of the U.S. Department of Energy under contract 89233218CNA000001 and supported by NIH grants R01-AI078881, R01-OD011095 (ASP), and R01-AI116868 (RMR). TC was supported by the Engineering and Physical Sciences Research Council [grant number 32061].

 \clearpage
 
 \appendix

\section{Generational tracking of infected hepatocytes}\label{Sec:InfiniteODEProof}

Here, we prove that the infinite system of ODEs in Eq.~\eqref{Eq:InfiniteSystemODE} has a strongly continuous solution by using the results of \citet{McClure1976}, who considered the generic infinite system of ODEs
\begin{align}\label{Eq:McClureODEAppendix}
\TimeDeriv y_i(t) = \displaystyle \sum_{j=1}^{\infty} a_{i,j}(t)y_j(t) + f_j(\tilde{y}(t)), \quad \textrm{for} \quad i = 1,2,3,....
\end{align}
By converting the infinite ODE~\eqref{Eq:McClureODEAppendix} to the equivalent system of integral equations, \citet{McClure1976} prove that
\begin{theorem}[Theorem 4.2 of \citep{McClure1976}] \label{Thm:McClure}
Assume that the linear portion of the right-hand side of Eq.~\eqref{Eq:McClureODEAppendix} satisfies
\begin{enumerate}
\item[\textbf{A1)}] Each linear coefficient $a_{i,j}(t)$ is continuous
\item[\textbf{A2)}] The function $w(t) = \sup\{ Re(a_{ij}(t)\}$ is locally bounded
\item[\textbf{A3)}] For each $j$, the sum $\displaystyle \sum_{i \neq j} |a_{i,j}(t)|$ converges uniformly on compact subsets of $\mathbb{R}$.
\end{enumerate}
Furthermore, assume that the nonlinear portion of Eq.~\eqref{Eq:McClureODEAppendix} satisfies
\begin{enumerate}
\item[\textbf{F1)}] Each function $f_i$ is continuous
\item[\textbf{F2)}] The sum $\|f \| = \displaystyle \sum_{i=1}^{\infty} |f_{i}(\tilde{y})|$ converges uniformly on bounded subsets of $\ell_1$
\item[\textbf{F3)}] There exists a continuous function $g(t)$ such that $|f (\tilde{y}) | \leq g(t) \|\tilde{y}\|_{\ell_1}$ for all time.
\end{enumerate}
Then, the infinite system of ODEs Eq.~\eqref{Eq:McClureODEAppendix} has a strongly continuous solution.
\end{theorem} 

We now return to the infinite system of ODEs in Eq.~\eqref{Eq:InfiniteSystemODE}, where the nonlinear portion of the infinite system of ODEs arises from the proliferation of infected hepatocytes.
 In developing the model of generational inheritance for infected cells, we explicitly distinguished between the proliferation of cells in generation $n-1$ to produce cells in generation $n$, and the loss of these cells in generation $n$ due to proliferation
\begin{align}\label{Eq:GenerationTerm}
        \TimeDeriv I_n(t) = 2 \eta b(t)  I_{n-1}(t) - \eta b(t) I_n(t) - \delta  I_n(t).
\end{align}
However, we can rewrite this equation to emphasize the net growth rate, $ \phi(t)  = \eta b(t) - \delta $, of infected cells. By adding and subtracting $\eta b(t)I_n$, Eq.~\eqref{Eq:GenerationTerm} becomes
\begin{align}\label{Eq:TransitCompartmentInfected}
    \TimeDeriv I_n(t) = 2\eta b(t) \left( I_{n-1}(t) -I_n(t) \right) + \phi(t) I_n(t), \quad \textrm{for} \quad n = 1,2,3,.... 
 \end{align}
The differential equation in Eq.~\eqref{Eq:TransitCompartmentInfected} is a transit compartment equation that models inflow and outflow from $n$-th generation at the nonconstant, but bounded, proliferation rate $2b(t)$, and captures the net growth of the population via the net growth rate  $\phi(t).$ Similar transit compartment ODEs occur frequently in applications of the linear chain technique \citep{Cassidy2020a,Champredon2018,Cassidy2018a}. 

Now, recall that newly infected hepatocytes must pass through the intervening $n$ generations before reaching the $n-$th generation. Progression through these generations is controlled by the proliferation rate, so the transit rate between generations is given by $2\eta b(t)$. Then, the concatenation of these $n$ generations, or stages, corresponds to an Erlang distribution with shape parameter $n$ for the time-lag from new infection to reaching generation $n$. These infected hepatocytes will undergo growth (or decay) depending on the \textit{net} growth rate $\phi(t)$. Thus, the concentration of infected hepatocytes in generation $n$ at the present is entirely determined by the concentration of newly infected hepatocytes in the past, and the net growth of these cells over the intervening time interval. Indeed, \citet{Cassidy2018a} prove that the solution of Eq.~\eqref{Eq:TransitCompartmentInfected} must be given by
\begin{equation}\label{Eq:InfectedCellSolution}
    I_n(t) = \int_0^{\infty} 2 \eta b(t-s) I_0(t-s) \exp \left[ \int_{t-s}^t \phi ( \sigma ) \d \sigma \right] g_{1}^n \left( \int_{t-s}^t 2\eta b(\sigma) \d \sigma \right) \d s,
\end{equation}
where 
\begin{align*}
    g_a^n(s) = \frac{a^n}{(n-1)!} s^{n-1}\exp(-a s )
\end{align*}
is the probability density function of the Erlang distribution with shape parameter $n$ and rate parameter $a$. The crux of this result is the relationship
\begin{align*}
    \frac{\d}{\d s} g_a^1(s) = - g_a^1(s) \quad \textrm{and} \quad  \frac{\d}{\d s} g_a^n(s) = a[g_a^{n-1}(s) - g_a^{n}(s)] \quad  \textrm{for} \quad n \geq 2,
\end{align*}
as using this equality and Leibniz's rule establishes Eq.~\eqref{Eq:InfectedCellSolution}. 
Consequently, both our biological intuition and Eq.~\eqref{Eq:InfectedCellSolution} link the concentration of infected cells in the $n$-th generation with newly infected cells in the past. This additional structure on the possible solutions of Eq.~\eqref{Eq:TransitCompartmentInfected} allows us to prove 
\begin{lemma}\label{Lemma:UniformSolutionInfiniteODE}
    Assume that proliferation rate $b(t)$ is bounded and the ODE Eq.~\eqref{Eq:InfiniteSystemODE} has a solution with $\tilde{I}(t) = (I_1(t), I_2(t),...).$ Then, the infinite series
    \begin{align*}
        \| \tilde{I} (t) \| =  \displaystyle \sum_{n=1}^{\infty} I_n(t)
    \end{align*}
    converges uniformly on compact time intervals.
    \end{lemma}
\begin{proof}
    We use the explicit expression for $I_n(t)$ in Eq.~\eqref{Eq:InfectedCellSolution} to explicitly calculate
    \begin{align*}
       \| \tilde{I}(t) \| =  \displaystyle \sum_{n=1}^{\infty} \int_0^{\infty} 2\eta b(t-s) I_0(t-s) \exp \left[ \int_{t-s}^t \phi(\sigma) \d \sigma \right] g_{1}^n \left(\int_{t-s}^t 2b(\sigma) \d \sigma \right) \d s.
    \end{align*}
    The integrand is strictly positive, so interchanging integration and summation using the Tonelli-Fubini theorem, we find
    \begin{align*}
       \| \tilde{I}(t) \ | =  \int_0^{\infty} 2\eta b(t-s) I_0(t-s) \exp \left[ \int_{t-s}^t \phi(\sigma) \d \sigma \right] \displaystyle \sum_{n=1}^{\infty}  g_{1}^n \left(\int_{t-s}^t 2b(\sigma) \d \sigma \right) \d s.
    \end{align*}
    Now,  the definition of $g_a^{n}(s)$ implies
    \begin{align*}
        \displaystyle \sum_{n=1}^{\infty}  g_{a}^n(s) = a \exp(-a s ) \displaystyle \sum_{n=1}^{\infty}  \frac{a^{n-1}  s^{n-1} }{(n-1)!} .
    \end{align*} 
    Of course, the latter series converges uniformly to $e^{as}$ and we find 
    \begin{align*}
        \displaystyle \sum_{n=1}^{\infty}  g_{a}^n(s) =  a. 
    \end{align*}
    It follows that the infinite series for $\| \tilde{I} \|$ converges uniformly to 
\begin{align*}
       \| \tilde{I}(t) \| =  \int_0^{\infty} 2\eta b(t-s) I_0(t-s) \exp \left[ \int_{t-s}^t \phi(\sigma) \d \sigma \right]  \d s.
    \end{align*} 
\end{proof}
In Lemma~\ref{Lemma:UniformSolutionInfiniteODE}, we assumed that the solution $\tilde{I}(t)$ existed. However, we have not yet shown that a solution to Eq.~\eqref{Eq:InfiniteSystemODE} exists. To do so, we need to show that Eq.~\eqref{Eq:InfiniteSystemODE} satisfies the assumptions of Theorem~\ref{Thm:McClure}.

Now, Theorem~\ref{Thm:McClure} considers closed and bounded subsets $z \subset \ell_1 \times \mathbb{R}.$ However, we have shown in Lemma~\ref{Lemma:UniformSolutionInfiniteODE} that the infinite series defined by any solution $\tilde{I}$ of the infinite ODE must be uniformly convergent. The space of sequences $\tilde{x} \in \ell_1$ that is closed, bounded, and where $\|\tilde{x}\| $ converges uniformly are precisely compact subsets $K \subset \ell_1$. Even when considering compact subsets $K$, establishing the uniform convergence in assumption \textbf{F2} can be delicate. We thus recall
\begin{theorem}[Dini's theorem]
    Let $K$ be a compact metric space and assume that $f:K \to \mathbb{R}$ is a continuous function. Further, assume that $f_n:K\to \mathbb{R}$ is a sequence of continuous functions with $f_{n-1}(x) \leq f_{n}(x)$ for all $x\in K$ and $n \in \mathbb{N}$. If $f_n(x)$ converges to $f(x)$ pointwise at each $x\in K$, then the convergence is uniform. 
\end{theorem}
Furthermore, for initial conditions satisfying $\|\tilde{I}(0)\| > 0$, Gronwall's inequality shows that we need only consider solutions with $I_n(t) \geq 0$ . 

\begin{proposition}
Assume that the parameters in Eq.~\eqref{Eq:InfiniteSystemODE} are strictly positive, the proliferation function $b(t)$ is continuous and bounded, and the initial conditions $\tilde{I}(t) = (I_1(0),I_2(0),...)$ are non-negative and belong to a compact subset $K \subset \ell_1$. Then, the infinite system of ODEs~Eq.~\eqref{Eq:InfiniteSystemODE} has a strongly continuous solution.
\end{proposition}

\begin{proof}
We need only to show that Eq.~\eqref{Eq:InfiniteSystemODE} satisfies the assumptions of Theorem~\ref{Thm:McClure} in the set $K \times I \subset \ell_1 \times \mathbb{R}$. In what follows, we consider the biologically relevant sequences $\tilde{x} \in K \subset \ell_1$, where $K$ is the compact set of all sequences with $x_i \geq 0$ for all $i$. Since $K$ is compact, the series $\| \tilde{x} \| $ converges uniformly.  

Now, the linear coefficients $a_{i,j}$ are constant in time, so \textbf{A1)} and \textbf{A2)} are immediately satisfied. Furthermore, for each $j$, the generational structure of Eq.~\eqref{Eq:InfiniteSystemODE} implies that only $a_{j-1,j}$ and $a_{j,j}$ are non-zero. Therefore, $\| a_{i,j}\| $ is a sum of only two non-zero terms and converges uniformly in any compact time interval.

We now turn to the conditions on the non-linear portion of Eq.~\eqref{Eq:InfiniteSystemODE}. Consider a bounded subset $Z \times I \subset \ell_1 \times \mathbb{R}$ with $\|\tilde{x}\| < T_{max}$ and take $\tilde{x} \in Z $. The nonlinear portion of Eq.~\eqref{Eq:InfiniteSystemODE} is given by  
\begin{align*}
f(\tilde{x}) &  =  \eta b(t,\tilde{x}) \left(  \sum_{n=0}^{\infty} x_n(t) \right) ,
\end{align*}
where we emphasize the dependence of the proliferation rate $b$ on the $\tilde{x}$. 

We now use Dini's theorem to show that the series defining $f(\tilde{x})$ converges uniformly on a compact subset $K \times I \subset \ell_1 \times \mathbb{R}$. We define 
\begin{align*}
f_M(\tilde{x}) = \displaystyle \sum_{i=1}^{M} \eta  b(t,\tilde{x})  x_i,
\end{align*} 
where we explicitly write the possible dependence of $b$ on $\tilde{x}$. Now, since $\tilde{x} \in K,$ the sum in the definition of $f_M(\tilde{x})$ is well-defined. Then, it is straightforward to see that the sequence of functions $\{ f_M(x) \}$ is increasing at any $x$. Furthermore, $f_{M}$ converges pointwise to $ f(\tilde{x})$.  

We now show that $f(\tilde{x})$ is continuous. Consider $\tilde{x},\tilde{y} \in K $ such that $\| \tilde{x} - \tilde{y} \| < h$, so we find
\begin{align*}
\left| f(\tilde{x}) - f(\tilde{y}) \right| & = \left|  \eta \displaystyle \sum_{i=1}^{\infty} \left( b(t,\tilde{x}) x_i - b(t,\tilde{y})  y_i \right)   \right|.
\end{align*}  
Regrouping terms and using the triangle inequality gives
\begin{align*}
\left| f(\tilde{x}) - f(\tilde{y}) \right| \leq \eta b_{max} \|x -y \|,
\end{align*}
where $b_{max}$ is the maximum of $b(t,\tilde{x})$. We immediately conclude that $f$ is continuous. A similar argument shows that each $f_M$ is continuous. Then, $f_M$ is an increasing sequence of continuous functions that converges pointwise to the continuous function $f$. As we are only considering compact subspaces of $\ell_1$, Dini's theorem then implies that $f_M$ converges to $f$ uniformly. Consequently, \textbf{F2} holds.
 
Finally, setting $g(t) =  b_{max} \,\eta \| $ immediately implies that \textbf{F3} is satisfied. Consequently, Theorem~\ref{Thm:McClure} applies to Eq.~\eqref{Eq:InfiniteSystemODE} and this system has a strongly continuous solution.
\end{proof}

\section{Initial age distributions during chronic infection}\label{Appendix:InitialDensities}

Often, models such as Eq.~\eqref{Eq:NoProliferationMultiScale} are used to understand viral kinetics in chronically infected participants in clinical trials of novel antivirals \citep{Rong2013a,Guedj2013}. We show how the assumption of chronic infection naturally defines the initial densities $f_{i}(a)$ and $f_{r}(a)$. These functions describe the density of infected hepatocytes with infection age $a^*>0$ at time $t=0$. These cells must have been infected at time $t=-a^*$, so we can project the characteristics of Eq.~\eqref{Eq:NoProliferationMultiScale} backwards and link the initial densities $f_a$ and $f_r$ to infections that occurred in the past \citep{Cassidy2020b,Cassidy2021}.

Mathematically, we assume that the system is in equilibrium prior to the beginning of therapy and therefore search for the stable age distributions of the PDEs describing the intracellular dynamics \citep{Rong2013a,Cassidy2021,Rong2013}. These distributions are constant in time and are given by the solution of 
\begin{align*}
\frac{\d}{\d a} i(a) & =   - \delta i(a), \\
\frac{\d}{\d a} r(a) & =  \alpha - \left( \mu +\rho \right)r(a).
\end{align*}
We thus obtain
\begin{align} \label{Eq:InfectedStableAge}
i(a) & =  i_0 e^{-\delta a},
\end{align}
and 
\begin{align} 
r(a) &  = \frac{\alpha}{\mu + \rho} + \left( \zeta - \frac{\alpha}{\mu + \rho} \right)e^{-(\mu + \rho) a} .\label{Eq:RNAIntraCellularStableAge}
\end{align}
 Thus, $r(a)$ is precisely the stable age distribution calculated by \citet{Rong2013} and is the natural definition for $f_r(a)$. Translating these stable age distributions to appropriate initial conditions for the system of ODEs obtained from multiscale PDE models can be delicate. However, we recall that the differential equation for the concentration of target cells does not arise from an underlying PDE model, so we can simply set $T(0) = T_0$ and $V(0) = V_0$. However, the variables $I(t)$ and $R(t)$ are defined via densities obtained from the corresponding age-structured PDE in Eq.~\eqref{Eq:NoProliferationMultiScale} along the characteristic lines. Therefore, their initial values are intrinsically linked to the initial age distribution of the intracellular variables. For example, it follows from the definition of $I(t)$ that $I(0)$ must be given by 
\begin{align*}
I_0 =\int_0^{\infty} i(0,a) \d a = \int_0^{\infty} f_i(a) \d a < \infty.
\end{align*}
We note that infected cells with age $ a^* > 0$ at time $t=0$ must have been infected at time $t= -a^*$ and that a fraction of these infected cells would die in the interim. Projecting the characteristics of \eqref{Eq:NoProliferationMultiScale} backwards from time $t=0$ \citep{Cassidy2020b}, the number of cells infected at time $t= -a^*$ must be given by
\begin{align*}
i(0-a^*,0) = i(0,a^*)e^{ \delta a^* } = f_i(a^*) e^{ \delta a^* },
\end{align*}
or alternatively, the number of infected cells that have died between infection and time $0$ is 
\begin{align*}
i(-a^*,0)- i(0,a^*) = i(-a^*,0) \left( 1- e^{-\delta a^*} \right),
\end{align*}
which corresponds precisely to the exponentially distributed lifespan assumed for infected cells. Now, we have assumed that the system is in a steady-state corresponding to chronic infection with viral load $V_0$ and the corresponding concentrations of uninfected hepatocytes,  $T_0$. The number of new infections at time $t = -a^*$ must therefore be $i(-a^*,0) = \beta V^* T^*$. As $a^*$ was arbitrary, the initial age distribution of infected cells is given by 
\begin{align*}
f_i(a) = i(-a,0)e^{-\delta a } =  \beta V^* T^* e^{-\delta a }. 
\end{align*} 
Thus, the initial number of infected cells is given by 
\begin{align*}
I(0) = \int_0^{\infty} f_i(a) \d a = \frac{\beta V^* T^*}{\delta}. 
\end{align*} 
Next, in chronic infection, we assume that the intracellular concentrations of vRNA have reached the stable age distribution given in Eq.~\eqref{Eq:RNAIntraCellularStableAge}, so 
\begin{align*}
R(0) & = \int_0^{\infty} r_0(a)i_0(a) \d a =  \int_0^{\infty} \frac{\alpha\beta V^* T^* }{\mu + \rho} e^{-\delta a} \d a  +  \int_0^{\infty}  \beta V^* T^*\left( \zeta - \frac{\alpha}{\mu + \rho} \right)e^{-(\mu + \rho+\delta) a}, \\
& =  \left( \frac{\alpha}{\mu + \rho} \right) \frac{\beta V^* T^*}{\delta} + \left( \beta V^* T^*\left( \zeta - \frac{\alpha}{\mu + \rho} \right)\right) \left( \frac{1}{\mu + \rho+\delta} \right).
\end{align*} 
 
\subsection{Generational distribution of infected hepatocytes in chronic infection}

In the preceding analysis, we showed how to calculate the initial density of infected hepatocytes, $f_i(a)$, using the stable age distribution of infected hepatocytes during chronic infection. This initial condition gives the total concentration of infected hepatocytes at time $t = 0$ without considering the generational structure of the infected cell population. We now show how the initial concentration of infected hepatocytes immediately defines the initial generational distribution of infected cells, $I_n(0)$, in Eq.~\eqref{Eq:InfiniteSystemODE}. 

We once again consider logistic proliferation and assume the initial concentration of uninfected and infected hepatocytes, $T(0)= T^*$ and $I(0) = I^*$, are known. Now,  we must have $T(0) + I(0) \leq T_{max}$ and the proliferation rate during chronic infection is
\begin{align*}
    b^* = \gamma \left(1 - \frac{T^* + I^*}{T_{max}} \right). 
\end{align*}
Now, we impose that the system is in equilibrium during chronic infection. Eq.~\eqref{Eq:GenerationTerm} yields
\begin{align*}
    I_n^* =  \frac{2\eta b^*}{\eta b^* + \delta} I^*_{n-1},
\end{align*}
which implies that 
\begin{align*}
    I_n^* = \left(  \frac{2\eta b^*}{\eta b^* + \delta} \right)^n I_0^*.
\end{align*}
Consequently, if we can calculate the initial concentration of infected cells in generation 0, we can obtain the initial concentration of infected hepatocytes for all later generations.

Now, the total concentration of infected hepatocytes at time $t= 0 $ is given by 
\begin{align}
    I^* =  I_0^* \displaystyle \sum_{n=0}^{\infty}\left(  \frac{2\eta b^*}{\eta b^* + \delta} \right)^n.
\end{align}
We immediately see that the initial concentration of infected hepatocytes is finite if and only if $2\eta b^* < \delta + \eta b^*,$ which is equivalent to $0 < \delta - \eta b^*.$ From this, we calculate
\begin{align*}
    I^* & = I_0 \left( \frac{\eta b^* + \delta}{\delta - \eta b^* } \right),
\end{align*}
which immediately yields $I_0$ (and thus $I_n(0)$) as $T^*$ and $I^*$, and thus $b^*$, are known.

\section{Saddle-node and transcritical bifurcations of the infected equilibria}\label{Appendix:BifurcationOfInfectedEquilibria}
In the main text, we showed that, at any infected equilibria that is not the total infection equilibria, i.e. with $T^* \in (0,T_{max}),$ the concentration of uninfected hepatocytes must satisfy $F(T^*) = 0$ where
\begin{align*}
F(T^*) = \gamma B T^* ( M(T^*)-B) + \gamma B\eta T_{max} (1-T^*/ T_{max}) M(T^*) + \delta T_{max}M(T^*)(M(T^*)-B), 
\end{align*}
is a quadratic function of $T^*$. Here, we study the possible solutions of $F(T^*) = 0$. We re-write $F(T^*)$ as
\begin{align}\label{Eq:FDefinitionAppendix}
F(T^*) = c_2 (T^*)^2 + c_1 T^* + c_0, 
\end{align}
where, after using the definitions of $B$, $\mathcal{R}^{\dagger}$, and adding and subtracting $\delta \eta \rho \beta T_{max}$ from the constant term, we find
\begin{equation}\label{Eq:QuadraticCoefficientsAppendix}
\left \{
\begin{aligned}
c_2 & =  \zeta \rho^2 \beta^2T_{max} ( \alpha (1-\eta) + \delta \zeta ), \\ \notag
c_1 & = T_{max} \zeta \rho \beta (\rho + \mu + \delta)c\delta \eta  \left(  \mathcal{R} - \mathcal{R}^{\dagger}- \frac{1}{\eta} \right) - (\rho + \mu + \delta)c \gamma B(1-\eta) -\gamma B^2, \\
c_0 & = \left[ (\rho + \mu + \delta)c \right]^2 \delta T_{max} \eta \left( \mathcal{R}^{\dagger}-\mathcal{R} \right) .
\end{aligned}
\right.
\end{equation}
First, we see that, if $\mathcal{R} = \mathcal{R}^{\dagger}$, then $c_0= 0$ and thus $F(0) = 0$. This root corresponds to the infected equilibrium coalesces with the total infection equilibrium $(0,I^{\dagger},R^{\dagger},V^{\dagger})$ in a transcritical bifurcation at $\mathcal{R} = \mathcal{R}^{\dagger}$.  

Now, as $F(T^*)$ is a quadratic function of $T^*$, it has exactly two roots. However, the equilibria defined by these roots may not be biologically relevant. The two roots of $F(T^*)$ are real-valued if and only if $c_1^2 - 4c_0c_2 > 0$, and these roots collide in a saddle-node bifurcation when $c_1^2 - 4c_0c_2 = 0$. 

We use the definition of $B$ to calculate
\begin{align*}
c_1^2 =\left[  \rho \beta T_{max} \left( \rho+\mu+ \delta \right)c\right]^2 \left[ \zeta \delta \eta \left( \mathcal{R}-\mathcal{R}^{\dagger} \right)- \left( \alpha(1-\eta) + \delta\zeta \right) - \frac{\alpha^2 \beta \rho T_{max}	}{\gamma (\rho+\mu+\delta)c} \right]^2,
\end{align*}
and
\begin{align*}
4c_0c_2 = 4\left[ \rho \beta T_{max}\left( \rho+\mu+\delta \right) c \right]^2 \zeta \delta \eta \left( \alpha (1-\eta) +\delta \zeta \right)\left( \mathcal{R}-\mathcal{R}^{\dagger} \right).
\end{align*}
Now, for notational simplicity, we denote 
\begin{align*}
k_1 = \zeta \delta \eta \left( \mathcal{R}-\mathcal{R}^{\dagger} \right), \quad k_2 =  \alpha(1-\eta) + \delta\zeta, \quad \textrm{and} \quad k_3 = \frac{\alpha^2 \beta \rho T_{max}	}{\gamma ( \rho+\mu+\delta)c}.
\end{align*}
Then, imposing $c_1^2 = 4c_0c_2$ and cancelling common terms gives
\begin{align*}
\left(k_1 -k_2 -k_3\right)^2 = 4k_1k_2.
\end{align*}
We can re-arrange and factor the resulting expression to find
\begin{align*}
k_1^2 -2k_1(3k_2+k_3) + (k_2+k_3)^2 = 0.
\end{align*}
Therefore, we have $c_1^2 - 4c_0c_2 = 0$, and thus a saddle-node bifurcation if both $k_1 = 0$ and $k_2 = -k_3$ are satisfied. The definitions of $k_{1}$ implies that $ k_1 = 0$ if and only if $ \mathcal{R}=\mathcal{R}^{\dagger}$, which gives the first condition in the main text. Conversely, the condition $k_2 = -k_3$ implies
\begin{align}\label{Eq:EtaSaddleNodeAppendix}
\eta = 1  + \frac{\delta\zeta}{\alpha} + \frac{ \alpha \beta \rho T_{max}	}{\gamma(\rho+\mu+\delta)c} ,
\end{align}
which is the second condition in the main text. We note that a saddle-node bifurcation can only occur if $\eta > 1,$ which corresponds to infected hepatocytes having a fitness advantage over uninfected hepatocytes.

\section{Infected equilibria if HCV infection induces a fitness cost}\label{Appendix:InfectedEquilibria}
In the preceeding analysis, we showed that a saddle-node bifurcation giving rise to two infected equilibria can only occur if infected hepatocytes have a fitness advantage over uninfected hepatocytes. Here, we show that if $\eta \leq 1,$ then it is not possible to have two infected equilibria, and thus rule out the co-existence of infected equilibria if infected hepatocytes do \textit{not} have a fitness advantage over uninfected hepatocytes. In short, we prove Lemma~\ref{Lemma:InfectedEquilibrium} in the main text. As before, we utilize the definition of $F(T^*)$ in Eq.~\eqref{Eq:FDefinitionAppendix}, with the coefficients $c_i$ given in Eq.~\eqref{Eq:QuadraticCoefficientsAppendix}. 

\begin{lemma}\label{Lemma:InfectedEquilibriumAppendix}
Let the model parameters be positive. Further, assume that 
\begin{align*}
\eta \in [0,1], \quad  \zeta \leq 1, \quad \textrm{and} \quad \alpha > \zeta \max \left[ 1 ,  1-\eta +\sqrt{(1-\eta)^2 + 4\delta} \right].  
\end{align*}
Eq.~\eqref{Eq:LogisticProliferationEquivalentODE} has an uninfected equilibrium $(T_{max},0,0,0)$ and the totally infected equilibrium $(0,I^{\dagger},R, V^{\dagger})$. In addition, Eq.~\eqref{Eq:LogisticProliferationEquivalentODE} has a unique infected equilibrium if and only if $1< \mathcal{R} < \mathcal{R}^{\dagger}$.
\end{lemma}

\begin{proof}
The condition on $\eta$ implies that $c_2 > 0$ and $\mathcal{R}^{\dagger} > 1$, while the condition on $\alpha$ implies that
\begin{align*}
\alpha^2 > 2\zeta (\alpha(1-\eta) + \zeta \delta).
\end{align*}
We distinguish between three cases. 

\textbf{Case 1)} Assume that $1 < \mathcal{R} < \mathcal{R}^{\dagger}$. It follows that $c_0, c_2 >0$ and hence $F(0) > 0$ while $c_1<0.$ Now, we return to Eq.~\eqref{Eq:TMaxTranscritical} and note that $\textrm{sign}(F(T_{max}) ) = \textrm{sign} ( 1- \mathcal{R}),$
since $\zeta < \alpha.$ Consequently, as $ \mathcal{R}>1$, we see that $F(T_{max}) < 0$. Recalling that $F$ is a quadratic function with positive leading co-efficient, the intermediate value theorem implies that $F$ has a unique root $T_{1}^* \in [0, T_{max}]$ that corresponds to an infected equilibrium. 

\textbf{Case 2)} Assume that $ 1 < \mathcal{R}^{\dagger} < \mathcal{R} $. It follows that $F(0) < 0$. Since $F$ has a positive leading coefficient, $F$ has a unique positive root $T_{2}^*$. However, $\mathcal{R}  > 1$, so Eq.~\eqref{Eq:TMaxTranscritical} implies that $F(T_{max}) < 0.$ Therefore, we must have $T_{2}^* > T_{max}$ and this equilibrium defined by this root is not biologically feasible.  
 
\textbf{Case 3)} Assume that $\mathcal{R} < 1 <  \mathcal{R}^{\dagger} $. It follows that both $c_0 > 0$ and $c_2 >0$. Then, Descartes' Rule of Signs indicates that $F$ has either two or zero positive roots.  Since $F(0) > 0$ and $F(T_{max}) > 0$, $F$ must have either 0 or two roots in the interval $(0,T_{max})$. We now show that these two positive roots, if they exist, do not define biologically feasible infected equilibria.

From the quadratic formula, the largest root $T_{+}^*$ satisfies
\begin{align*}
T_{+}^* \geq \frac{-c_1}{2c_2} = 
\frac{T_{max} \zeta \rho \beta (\rho + \mu + \delta)c\delta \eta  \left( \mathcal{R}^{\dagger} + \frac{1}{\eta}  - \mathcal{R} \right) + (\rho + \mu + \delta)c \gamma B(1-\eta) +\gamma B^2}{2 \zeta \rho^2 \beta^2T_{max} ( \alpha (1-\eta) + \delta \zeta )}.
\end{align*} 
Utilizing the condition $\mathcal{R} < \mathcal{R}^{\dagger}$ and the definition of $B$, we find
\begin{align*}
 T_{+}^* > \frac{  (\rho + \mu + \delta)c\delta }{ 2  \rho  \beta ( \alpha (1-\eta) + \delta \zeta ) }  + \frac{\alpha (\rho + \mu + \delta)c (1-\eta)}{2 \zeta \rho \beta ( \alpha (1-\eta) + \delta \zeta )} +   T_{max} \frac{\alpha }{2 \zeta  ( 1-\eta + \frac{\delta \zeta}{\alpha} )} . 
\end{align*} 
Recall that $\eta,\zeta \in [0,1]$ and $ \alpha > 2  ( 1 + \frac{\delta}{\alpha} )$ by the assumptions of the Lemma. Then $T_{+} > T_{max}$, so we conclude that $F(T^*)$ has no roots in the interval $(0,T_{max})$ in the case $\mathcal{R} < 1 < \mathcal{R}^{\dagger}$.
\end{proof}
 
\section{Including direct acting antiviral treatment in the multiscale model}\label{Appendix:Treatment}

\citet{Guedj2013} incorporated the direct acting antiviral effects of NS5A inhibitor treatment in Eq.~\eqref{Eq:NoProliferationMultiScale}. There, \citet{Guedj2013} included three distinct antiviral effects of NS5A inhibitor treatment, namely decreasing the production rate of intracellular vRNA by a factor $1-\epsilon_{\alpha}$, where $\epsilon_{\alpha} \in [0,1]$ is the drug effect on the production of intracellular vRNA, increasing the degradation rate of intracellular vRNA by a factor $\kappa \geq 1$, or decreasing the viral assembly and secretion rate by a factor $1-\epsilon_{s}$, where $\epsilon_s \in [0,1]$. For simplicity, the antiviral effects of the NS5A inhibitor were assumed to be constant during treatment \citep{Guedj2013,Rong2013a}. We give a schematic representation of the intracellular effect of NS5A inhibitor treatment in Fig.~\ref{Fig:TreatmentAppendix}.
\begin{figure}[htbp!]
\noindent
\centering 
\includegraphics[trim=  10  10 10 10,clip,width=0.95\textwidth]{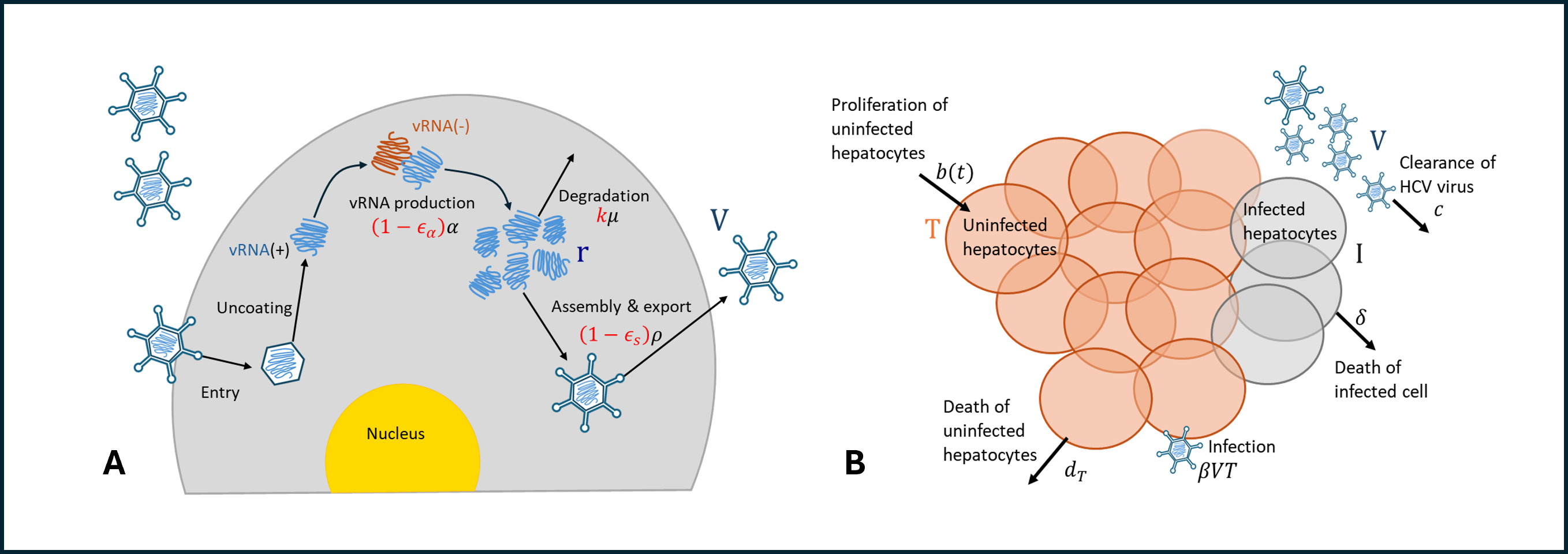}  
\caption{ \textbf{Schematic of the multiscale model of HCV infection with direct acting antiviral treatment}.  Panel A shows the intracellular HCV life cycle, which begins with infection of a hepatocyte and the release of positive strand vRNA, shown in blue and denoted vRNA(+), into the cell cytoplasm following the uncoating of the HCV viral capsid. This vRNA is translated into viral proteins, such as the RNA-dependent RNA polymerase and other proteins that first copy the positive strand vRNA into its complementary negative strand, shown in orange and denoted vRNA(-), and then form a replication complex or replication machine that generates new positive stand vRNA at rate $\alpha$, which can be degraded at rate $\mu$, or assembled into HCV particles and exported into the circulation at rate $\rho$.  Treatment with a direct acting antiviral, such as daclatasvir, that interacts with the HCV NS5A protein, inhibits the production of vRNA with an effectiveness of $\epsilon_{\alpha}$, inhibits the secretion of mature virus with effectiveness $\epsilon_s$ and increases the degradation of virus complexes by a factor $\kappa$. Panel B shows the HCV extracellular dynamics, where uninfected hepatocytes $(T)$ are produced due to proliferation at rate $b(t)$, and become infected cells $(I)$, following infection by virus $(V)$ at rate $\beta$. Infected cells are lost at per capita rate $\delta$ and secrete vRNA containing particles as virus ($V$) into the circulation. The circulating virus is cleared at per capita rate $c$.  }
\label{Fig:TreatmentAppendix}
\end{figure}
Under these assumptions, the multiscale model incorporating treatment is given by
\begin{equation} \label{Eq:TreatedNoProliferationMultiScale}
\left \{
\begin{aligned}
\TimeDeriv T(t) & = \lambda - d_T T(t) -\beta V(t) T(t), \\
(\partial_t+ \partial_a)i(t,a) & =  - \delta i(t,a), \\
(\partial_t+ \partial_a)r(t,a) & = \alpha(1-\epsilon_{\alpha})  -(\kappa \mu + \rho (1-\epsilon_{s}) ) r(t,a), \\
\TimeDeriv V(t) & = \int_{0}^{\infty} \rho (1-\epsilon_{s}) r(t,a)i(t,a) \d a  - cV(t).
\end{aligned}
\right.
\end{equation}
Then, \citet{Kitagawa2018} showed that this system of PDEs is equivalent to the ODE
\begin{equation} \label{Eq:NoProliferationODEAppendix}
\left \{
\begin{aligned}
\TimeDeriv T(t) & = \lambda - d_T T(t) -\beta V(t) T(t), \\
\TimeDeriv I(t) & = \beta V(t)T(t) - \delta I(t), \\
\TimeDeriv R(t) & = \zeta \beta V(t) T(t) +  \alpha (1-\epsilon_{\alpha}) I(t)  - (\rho (1-\epsilon_s)+ \kappa \mu+\delta)R(t), \\
\TimeDeriv V(t) & =  \rho R(t)  - cV(t).
\end{aligned}
\right.
\end{equation}
It is straightforward to include these antiviral effects in the models we developed throughout the main text. 

\section{Spontaneous generation of vRNA during proliferation in a multiscale model of HCV infection with proliferation}\label{Appendix:Comparison}

As previously mentioned, \citet{Elkaranshawy2024} recently adapted the multiscale model of HCV infection to include proliferation of infected hepatocytes. There, they directly worked with the equivalent ODE formulation found by \citet{Kitagawa2018} rather than the multiscale model capturing the intracellular viral life cycle as in Eq.~\eqref{Eq:ProliferationMultiScale}. After including the effects of antiviral NS5A treatment as in \citet{Guedj2013}, \citet{Elkaranshawy2024} obtained
\begin{equation} \label{Eq:ElkaranshawyProliferationODE}
\left \{
\begin{aligned}
\TimeDeriv T(t) & = \lambda + \gamma \left(1-\frac{T(t)+I(t)}{T_{max}} \right) T(t) -d_T T(t) - \beta V(t) T(t), \\
\TimeDeriv I(t) & = \beta V(t)T(t) +  \gamma \left(1-\frac{T(t)+I(t)}{T_{max}} \right)  I(t)  - \delta I(t), \\
\TimeDeriv R(t) & = \zeta \left( \beta V(t) T(t) +  \gamma \left(1-\frac{T(t)+I(t)}{T_{max}} \right)  I(t)  \right) +  \alpha( 1-\epsilon_{\alpha} ) I(t) \\
& {} \quad - (\rho (1-\epsilon_s )+\kappa \mu+\delta) R(t), \\
\TimeDeriv V(t) & =  \rho R(t)  - cV(t).
\end{aligned}
\right.
\end{equation}
The model in Eq.~\eqref{Eq:ElkaranshawyProliferationODE} differs only from our model Eq.~\eqref{Eq:LogisticProliferationEquivalentODE} by the inclusion of the logistic growth term in the differential equation for the total amount of intracellular vRNA within infected hepatocytes. We recall that $\zeta$ represents the amount of intracellular vRNA within newly infected hepatocytes due to infection. Consequently, this term models an increase in the intracellular vRNA within daughter cells of an infected hepatocyte. Thus, the second term in the ODE for $R(t)$ in   Eq.~\eqref{Eq:ElkaranshawyProliferationODE} corresponds to the generation of intracellular vRNA due to cellular proliferation. However, as we have already discussed, proliferation of an infected hepatocyte should, at most, conserve the total amount of intracellular vRNA, rather than generating new vRNA. Finally, we note that this spontaneous generation of vRNA upon proliferation would arise in our PDE model if we did \textit{not} distinguish between the different biological processes that lead to newly infected hepatocytes, $i_0(t,a)$, and those that lead to the production of infected hepatocytes arising from infected cell proliferation, $i_p(t,a)$.

\end{document}